\pdfoutput=1

\documentclass[a4paper,USenglish]{lipics-v2016-modified}

\usepackage{style/allLIPICS}

\newcommand{\horsarity}[1]{\mathsf{arity}\left(#1\right)}
\newcommand{\horsorder}[1]{\mathsf{order}\left(#1\right)}

\newcommand{\ground}{o}

\newcommand{\contextplaceholder}{\bullet}
\newcommand{\abs}{\alpha}
\newcommand{\absof}[1]{\abs(#1)}
\newcommand{\concsem}[1]{\semunder{\modelconc}{#1}}

\newcommand{\disunion}{\dotcup}

\renewcommand{\sol}[1]{\sigma_{#1}}

\newcommand{\soliter}[2]{\sigma^{#2}_{#1}}
\newcommand{\solconc}{\sol{\modelconc}}
\newcommand{\solabs}{\sol{\modelabs}}

\newcommand{\solconcof}[1]{\solconc(#1)}

\newcommand{\wordend}{\$}
\newcommand{\QNFA}{Q_{\mathit{NFA}}}

\newcommand{\acceptfrom}[1]{\mathsf{acc}\left(#1\right)}
\newcommand{\acceptfromf}{\mathsf{acc}}
\newcommand{\predec}[2]{\mathsf{pre}_{#1}\left(#2\right)}

\newcommand{\prepend}[2]{\mathsf{prepend}_{#1}(#2)}
\newcommand{\predecf}[1]{\mathsf{pre}_{#1}}
\newcommand{\prependf}[1]{\mathsf{prepend}_{#1}}

\newcommand{\prightarrow}{\nrightarrow}

\newcommand{\soundfor}{\vdash}

\newcommand{\conctag}{\mathit{C}}
\newcommand{\abstag}{\mathit{A}}
\newcommand{\opttag}{\mathit{O}}
\newcommand{\meet}{\sqcap}

\newcommand{\bigmeet}{\bigsqcap}
\newcommand{\inter}{\mathcal{I}}
\newcommand{\interof}[1]{\inter(#1)}
\newcommand{\interconc}{\mathcal{I}^{\conctag}}
\newcommand{\interconcof}[1]{\inter^{\conctag}(#1)}

\newcommand{\interabsof}[1]{\inter^{\abstag}(#1)}
\newcommand{\interopt}{\mathcal{I}^{\opttag}}
\newcommand{\interoptof}[1]{\inter^{\opttag}(#1)}

\newcommand{\contfun}[2]{\mathit{Cont}(#1, #2)}
\newcommand{\model}{\mathcal{M}}
\newcommand{\modelconc}{\model^{\conctag}}
\newcommand{\modelabs}{\model^{\abstag}}
\newcommand{\modelopt}{\model^{\opttag}}
\newcommand{\semunder}[2]{#1\sem{#2}}
\newcommand{\domain}{\mathcal{D}}

\newcommand{\domainof}[1]{\domain({#1})}
\newcommand{\domainconc}{\mathcal{D}^{\conctag}}
\newcommand{\domainconcof}[1]{\domainconc(#1)}
\newcommand{\domainabs}{\mathcal{D}^{\abstag}}
\newcommand{\domainabsof}[1]{\domainabs(#1)}
\newcommand{\domainopt}{\mathcal{D}^{\opttag}}
\newcommand{\domainoptof}[1]{\domainopt(#1)}
\newcommand{\rhsfun}[1]{\mathit{rhs}_{#1}}
\newcommand{\rhsfunof}[2]{\rhsfun{#1}(#2)}

\newcommand{\precisioncont}{(\mathsf{P2})}
\newcommand{\precisionsur}{(\mathsf{P1})}
\newcommand{\precisionbot}{(\mathsf{P3})}
\newcommand{\precisionrel}{(\mathsf{P4})}
\newcommand{\precisioncomp}{(\mathsf{P5})}
\newcommand{\opof}[1]{\mathit{op}_{#1}}
\newcommand{\detof}[1]{#1^{\mathit{det}}}

\newcommand\modell{\model_l}
\newcommand\modelr{\model_r}
\newcommand\domainl{\domain_l}
\newcommand\domainr{\domain_r}
\newcommand\domainlof[1]{\domainl(#1)}
\newcommand\domainrof[1]{\domainr(#1)}
\newcommand\interl{\inter_l}
\newcommand\interr{\inter_r}
\newcommand\interlof[1]{\interl(#1)}
\newcommand\interrof[1]{\interr(#1)}
\newcommand\vall{v_l}
\newcommand\valr{v_r}
\newcommand\toplof[1]{\top^l_{#1}}
\newcommand\toprof[1]{\top^r_{#1}}
\newcommand\fl{f_l}
\newcommand\fr{f_r}
\newcommand\solfl{\sol{l}}
\newcommand\solfr{\sol{r}}

\newcommand\rhsfunl{\rhsfun{\modell}}
\newcommand\rhsfunr{\rhsfun{\modelr}}
\newcommand\rhsfunopt{\rhsfun{\modelopt}}

\newcommand\genop{\mathop{\mathit{op}}}

\newcommand\gamehorsreg{\mathsf{HOG}}
\newcommand\playere{\Diamond}
\newcommand\playera{\Box}

\newcommand\horsgame{\mathcal{G}}

\newcommand\boolof[1]{\mathsf{PBool}\mathord{\left(#1\right)}}
\newcommand\ftrue{\mathsf{true}}

\newcommand\fmla{\phi}

\newcommand\soln{\Xi}
\newcommand\solset{\Omega}

\newcommand\ldisj[1]{\bigvee{}_{\!#1}}

\newcommand{\textsfbf}[1]{\textbf{\sffamily #1}}

\newcommand{\setfmla}{\Phi}

\usepackage{xspace}
\renewcommand{\ProblemLocDecide}{Question}

\newcommand\heveaornot[2]{#2}

\title
{
    Domains for Higher-Order Games
}

\author[1]{Matthew Hague}
\author[2]{Roland Meyer\footnote{A part of the work was carried out when the author was at Aalto University.}}
\author[2]{Sebastian Muskalla}

\affil[1]
{
    Royal Holloway University of London, United Kingdom\\
    \texttt{matthew.hague@rhul.ac.uk}
}

\affil[2]
{
    TU Braunschweig, Germany\\
    \texttt{\{roland.meyer, s.muskalla\}@tu-braunschweig.de}
}

\authorrunning{
    M. Hague, R. Meyer, and S. Muskalla
}

\Copyright
{
    Matthew Hague, Roland Meyer, and Sebastian Muskalla
}



\EventEditors{Kim G. Larsen, Hans L. Bodlaender, and Jean-Francois Raskin}
\EventNoEds{3}
\EventLongTitle{42nd International Symposium on Mathematical Foundations of Computer Science (MFCS 2017)}
\EventShortTitle{MFCS 2017}
\EventAcronym{MFCS}
\EventYear{2017}
\EventDate{August 21--25, 2017}
\EventLocation{Aalborg, Denmark}
\EventLogo{}
\SeriesVolume{83}
\ArticleNo{59}

\begin{document}

\maketitle

\begin{abstract}
We study two-player inclusion games played over word-generating higher-order recursion schemes.
While inclusion checks are known to capture verification problems, two-player games generalize this relationship to program synthesis.
In such games, non-terminals of the grammar are controlled by opposing players.
The goal of the existential player is to avoid producing a word that lies outside of a regular language of safe words.

We contribute a new domain that provides a representation of the winning region of such games.
Our domain is based on (functions over) potentially infinite Boolean formulas with words as atomic propositions.
We develop an abstract interpretation framework that we instantiate to abstract this domain into a domain where the propositions are replaced by states of a finite automaton.
This second domain is therefore finite and we obtain, via standard fixed-point techniques, a direct algorithm for the analysis of two-player inclusion games.
We show, via a second instantiation of the framework, that our finite domain can be optimized, leading to a $(k+1)\mathsf{EXP}$ algorithm for order-$k$ recursion schemes.
We give a matching lower bound, showing that our approach is optimal.
Since our approach is based on standard Kleene iteration, existing techniques and tools for fixed-point computations can be applied.
\end{abstract}

\section{Introduction}
Inclusion checking has recently received considerable attention \cite{Wulf2006,FogartyVardi,Abdulla:Simulation,Abdulla:Advanced,NETYS}.
One of the reasons is a new verification loop, which invokes inclusion as a subroutine in an iterative fashion.
The loop has been proposed by Podelski et al.\ for the safety verification of recursive programs~\cite{HeizmannHoenickePodelski2010}, and then been generalized to parallel and parameterized programs~\cite{LanguageRefinement,FKP14,FKP15} and to liveness~\cite{FKP16}.
The idea of Podelski's loop is to iteratively approximate unsound data flow in the program of interest, and add the approximations to the specification.
Consider a program with control-flow language $\mathit{CF}$ that is supposed to satisfy a safety specification given by a regular language $R$.
If the check $\mathit{CF}\subseteq R$ succeeds, then the program is correct as the data flow only restricts the set of computations.
If a computation $w\in \mathit{CF}$ is found that lies outside $R$, then it depends on the data flow whether the program is correct.
If data is handled correctly, $w$ is a counterexample to $R$.
Otherwise, $w$ is generalized to a regular language $S$ of infeasible computations.
We set $R = R \cup S$ and repeat the procedure.

Podelski's loop has also been generalized to synthesis~\cite{HMM16,MMN17}.
In that setting, the program is assumed to have two kinds of non-determinism.
Some of the non-deterministic transitions are understood to be controlled by the environment.
They provide inputs that the system has to react to, and are also referred to as demonic non-determinism.
In contrast, the so-called angelic non-determinism are the alternatives of the system to react to an input.
The synthesis problem is to devise a controller that resolves the angelic non-determinism in a way that a given safety specification is met.
Technically, the synthesis problem corresponds to a two-player perfect information game, and the controller implements a winning strategy for the system player.
When generalizing Podelski's loop to the synthesis problem, the inclusion check thus amounts to solving a strategy-synthesis problem.

Our motivation is to synthesize functional programs with Podelski's loop.
We assume the program to be given as a non-deterministic higher-order recursion scheme where the non-terminals are assigned to two players.
One player is the system player who tries to enforce the derivation of words that belong to a given regular language.
The other player is the environment, trying to derive a word outside the language.
The use of the corresponding strategy-synthesis algorithm in Podelski's loop comes with three characteristics:
(1) The algorithm is invoked iteratively,
(2) the program is large and the specification is small, and
(3) the specification is non-deterministic.
The first point means that the strategy synthesis should not rely on costly precomputation.
Moreover, it should have the chance to terminate early.
The second says that the cost of the computation should depend on the size of the specification, not on the size of the program.
Computations on the program, in particular iterative ones, should be avoided.
Together with the third characteristic, these two consequences rule out reductions to reachability games.
The required determinization would mean a costly precomputation, and the reduction to reachability would mean a product with the program.
This discussion in particular forbids a reduction of the strategy-synthesis problem to higher-order model checking~\cite{O06}, which indeed can be achieved (see Appendix~\ref{Appendix:Introduction} for a comparison to intersection types~\cite{KO09}).
Instead, we need a strategy synthesis that can directly deal with non-deterministic specifications.

We show that the winning region of a higher-order inclusion game \wrt a non-deterministic right-hand side can be computed with a standard fixed-point iteration.
Our contribution is a domain suitable for this computation.
The key idea is to use Boolean formulas whose atomic propositions are the states of the targeted finite automaton.
While a formula-based domain has recently been proposed for context-free inclusion games~\cite{HMM16} (and generalized to infinite words~\cite{MMN17}), the generalization to higher-order is new.
Consider a non-terminal that is ground and for which we have computed a formula.
The Boolean structure reflects the alternation among the players in the plays that start from this non-terminal.
The words generated along the plays are abstracted to sets of states from which these words can be accepted.
Determining the winner of the game is done by evaluating the formula when sets of states containing the initial state are assigned the value true.
To our surprise, the above domain did not give the optimal complexity.
Instead, it was possible to further optimize it by resolving the determinization information.
Intuitively, the existential player can also resolve the non-determinism captured by a set.
Crucially, our approach handles the non-determinism of the specification inside the analysis, without preprocessing.

Besides offering the characteristics that are needed for Podelski's loop, our development also contributes to the research program of \emph{effective denotational semantics}, as recently proposed by Salvati and Walukiewicz~\cite{SW15} as well as Grellois and Melli\`es~\cite{GM15MFCS,GM15MFCS}, with
\cite{A07,Salvati09} being early works in this field.
The idea is to solve verification problems by computing the semantics of a program in a suitable domain.
Salvati and Walukiewicz studied the expressiveness of greatest fixed-point semantics and their correspondence to automata~\cite{SW15}, and constructions of enriched Scott models for parity conditions~\cite{SW15c,SW15b}.
A similar line of investigation has been followed in recent work by Grellois and Melli\`es~\cite{GrelloisMelliesLinear15,GrelloisMelliesRelational15}.
Hofmann and Chen considered the verification of more restricted $\omega$-path properties with a focus on the domain~\cite{HC14}.
They show that explicit automata constructions can be avoided and give a domain that directly captures subsets (so-called patches) of the $\omega$-language.
The work has been generalized to higher order~\cite{HL17}.
Our contribution is related in that we focus on the domain (suitable for capturing plays).

Besides the domain, the correctness proof may be of interest.
We employ an exact fixed-point transfer result as known from abstract interpretation.
First, we give a semantic characterization showing that the winning region can be captured by an infinite model (a greatest fixed point).
This domain has as elements (potentially infinite) sets of (finite) Boolean formulas.
The formulas capture plays (up to a certain depth) and the atomic propositions are terminal words.
The infinite set structure is to avoid infinite syntax.
Then we employ the exact fixed-point transfer result to replace the terminals by states and get rid of the sets.
The final step is another exact fixed-point transfer that justifies the optimization.
We give a matching lower bound.
The problem is $(k+1)\mathsf{EXP}$-complete for order-$k$ schemes.


%

\subparagraph*{Related Work.}
The relationship between recursion schemes and extensions of pushdown automata has been well studied~\cite{Damm82,DG86,KNU02,HMOS08}.
This means algorithms for recursion schemes can be transferred to extensions of pushdown automata and vice versa.
In the sequel, we will use \emph{pushdown automata} to refer to pushdown automata and their family of extensions.

The decidability of Monadic Second Order Logic (MSO) over trees generated by recursion schemes was first settled in the restricted case of \emph{safe} schemes by Knapik~\textit{et al.}~\cite{KNU02} and independently by Caucal~\cite{Caucal02}.
This result was generalized to all schemes by Ong~\cite{O06}.
Both of these results consider \emph{deterministic} schemes only.

Related results have also been obtained in the consideration of games played over the configuration graphs of pushdown automata~\cite{Walukiewicz2001234,Cachat03,KNUW05,HMOS08}.
Of particular interest are \emph{saturation} methods for pushdown games~\cite{BEM97,FWW97,Cachat2002,BM04,HO07,HO09,BCHS12}.
In these works, automata representing sets of winning configurations are constructed using fixed-point computations.

A related approach pioneered by Kobayashi~et al.\ operating directly on schemes is that of \emph{intersection types}~\cite{KobayashiPOPL09,KO09}, where types embedding a property automaton are assigned to terms of a scheme.
Recently, saturation techniques were transferred to intersection types by Broadbent and Kobayashi~\cite{BK13}.
The typing algorithm is then a least fixed-point computation analogous to an optimized version of our Kleene iteration, restricted to deterministic schemes.
This has led to one of the most competitive model-checking tools for schemes~\cite{horsat2}.

One may reduce our language inclusion problems to many of the above works.
E.g.\ from an inclusion game for schemes, we may build a game over an equivalent kind of pushdown automaton and take the product with a determinization of the NFA.
This obtains a reachability game over a pushdown automaton that can be solved by any of the above methods.
However, such constructions are undesirable for iterative invocations as in Podelski's~loop.

%

We already discussed the relationship to model-theoretic verification algorithms.
Abstract interpretation has also been used by Ramsay~\cite{R14}, Salvati and Walukiewicz~\cite{SW15c,SW15b}, and Grellois and Melli\`es~\cite{GM15MFCS,G16} for verification.
The former used a Galois connection between safety properties (concrete) and equivalence classes of intersection types (abstract) to recreate decidability results known in the literature.
The latter two strands gives a semantics capable of computing properties expressed in MSO.
Indeed, abstract interpretation has long been used for static analysis of higher-order programs~\cite{AH87}.

\subparagraph*{Acknowledgments.}

This work was supported by the Engineering and Physical Sciences Research Council [EP/K009907/1].
The work instigated while some of the authors were visiting the Institute for Mathematical Sciences, National University of Singapore in 2016.
The visit was partially supported by the Institute.


\section{Preliminaries}

\subparagraph*{Complete Partial Orders.}

Let $(D, \leq)$ be a \emph{partial order} with set $D$ and (partial) ordering $\leq$ on $D$.
We call $(D, \leq)$ \emph{pointed} if there is a greatest element, called the \emph{top element} and denoted by $\top \in D$.
A \emph{descending chain} in $D$ is a sequence $(d_i)_{i \in \N}$ of elements in $D$ with $d_i \geq d_{i+1}$.
We call $(D, \leq)$ \emph{$\omega$-complete} if every descending chain has a greatest lower bound, called the \emph{meet} or the \emph{infimum}, and denoted by $\bigsqcap_{i \in  \N} d_i$.
If $(D, \leq)$ is pointed and $\omega$-complete, we call it a \emph{pointed $\omega$-complete partial order (cppo)}.
In the following, we will only consider partial orders that are cppos.
Note, cppo is usually used to refer to the dual concept, \ie partial orders with a least element and least upper bounds for ascending chains.



A function $f : D \to D$ is \emph{$\meet$-continuous} if for all descending chains $(d_i)_{i \in \N}$ we have $f(\bigsqcap_{i \in \N} d_i) = \bigsqcap_{i \in \N} f(d_i)$.
We call a function $f : D \to D$ \emph{monotonic} if for all $d, d' \in D$, $d \leq d'$ implies $f(d) \leq f(d')$.
Any function that is $\meet$-continuous is also monotonic.
For a monotonic function,
$
    \top \geq f(\top) \geq f^2 (\top) = f(f(\top)) \geq f^3 (\top) \geq \ldots
$
is a descending chain.

If the function is $\meet$-continuous, then $\bigsqcap_{i \in  \N} f^i (\top)$ is by Kleene's theorem the greatest fixed point of $f$, \mbox{\ie $f(\bigsqcap_{i \in  \N} f^i (\top)) = \bigsqcap_{i \in  \N} f^i (\top)$} and $\bigsqcap_{i \in  \N} f^i (\top)$ is larger than any other element $d$ with $f(d) = d$.
We also say $\bigsqcap_{i \in  \N} f^i (\top)$ is the greatest solution to the equation $x = f(x)$.

A lattice satisfies the \emph{descending chain condition (DCC)} if every descending chain has to be stationary at some point.
In this case $\bigsqcap_{i \in  \N} f^i (\top) = \bigsqcap_{i = 0}^{i_0} f^i (\top)$ for some index $i_0$ in $\N$.
With this, we can compute the greatest fixed point:
Starting with $\top$, we iteratively apply $f$ until the result does not change.
This process is called \emph{Kleene iteration}.
Note that finite cppos, \ie with finitely many elements in $D$, trivially satisfy the descending chain condition.


\subparagraph{Finite Automata.}

A \emph{non-deterministic finite automaton (NFA)} is a tuple
\mbox{$A = (\QNFA, \Gamma, \delta, q_0, Q_f)$}
where
    $\QNFA$ is a finite set of states,
    $\Gamma$ is a finite alphabet,
    $\delta \subseteq \QNFA \times \Gamma \times \QNFA$ is a (non-deterministic) transition relation,
    $q_0 \in \QNFA$ is the initial state, and
    $Q_f \subseteq \QNFA$ is a set of final states.
We write
$q \tow{a} q'$
to denote
$(q, a, q') \in \delta$.
Moreover, given a word
$w = a_1\cdots a_\ell$,
we write
$q \tow{w} q'$
whenever there is a sequence of transitions, also called \emph{run},
$q_1 \tow{a_1} q_2 \tow{a_2} \cdots \tow{a_\ell} q_{\ell+1}$
with $q_1 = q$ and $q_{\ell+1} = q'$.
The run is accepting if $q = q_0$ and $q' \in Q_f$.
The language of $A$ is
$
    \lang{A} = \Set{w}{q_0 \tow{w} q \in Q_f} \ .
$

\section{Higher-Order Recursion Schemes}
\label{Section:HORS}

We introduce higher-order recursion schemes, \emph{schemes} for short, following the presentation in~\cite{HaddadIOvsOI}.
Schemes can be understood as grammars generating the computation trees of programs in a functional language.
As is common in functional languages, we need a typing discipline.
To avoid confusion with type-based approaches to higher-order model checking~\cite{KobayashiPOPL09,R13,KO09}, we refer to types as \emph{kinds}.
Kinds define the functionality of terms, without specifying the data domain.
Technically, the only data domain is the ground kind $\ground$, from which (potentially higher-order) function kinds are derived by composition:
\begin{align*}
    \kappa\ ::=\ \ground \ \mid \ ( \kappa_1 \to \kappa_2 )\ .
\end{align*}
We usually omit the brackets and assume that the arrow associates to the right.
The number of arguments to a kind is called the \emph{arity}.
The \emph{order} defines the functionality of the arguments: A first-order kind defines functions that act on values, a second-order kind functions that expect functions as parameters.
Formally, we have
\begin{align*}
        \horsarity{\ground} &= 0,
        & \horsorder{\ground} &= 0,\\
        \horsarity{\kappa_1 \to \kappa_2} &= \horsarity{\kappa_2}+1,
        & \horsorder{\kappa_1 \to \kappa_2} &= \max(\horsorder{\kappa_1} +1, \horsorder{\kappa_2})\ .
\end{align*}
Let $\Kappa$ be the set of all kinds.
Higher-order recursion schemes assign kinds to symbols from different alphabets, namely non-terminals, terminals, and variables.
Let $\Gamma$ be a set of such \emph{kinded symbols}.
For each kind $\kappa$, we denote by $\Gamma^\kappa$ the restriction of $\Gamma$ to the symbols with kind $\kappa$.
The \emph{terms} $\calT^\kappa (\Gamma)$ of kind $\kappa$ over $\Gamma$ are defined by simultaneous induction over all kinds.
They form the smallest set satisfying
\vspace*{-0.3cm}
\begin{enumerate}
\item
    $\Gamma^\kappa \subseteq \calT^\kappa (\Gamma)$,
\item
    $\bigcup_{\kappa_1} \Set{t\ v}{t \in \calT^{\kappa_1 \to \kappa_2}(\Gamma), v \in \calT^{\kappa_1} (\Gamma)} \subseteq \calT^{\kappa_2} (\Gamma)$, and
\item
    $\Set{\lambda x.t}{x\in \calT^{\kappa_1}(\Gamma), t\in \calT^{\kappa_2}(\Gamma)}\subseteq  \calT^{\kappa_1 \to \kappa_2}(\Gamma)$.
\end{enumerate}
If term $t$ is of kind $\kappa$, we also write $t \colon \kappa$.
We use $\calT(\Gamma)$ for the set of all terms over $\Gamma$.
We say a term is \emph{$\lambda$-free} if it contains no sub-term of the form $\lambda x . t$.
A term is \emph{variable-closed} if all occurring variables are bound by a preceding $\lambda$-expression.
\begin{definition}
    A \emph{higher-order recursion scheme}, (\emph{scheme} for short), is a tuple $G = (V, N, T, R, S)$, where
    $V$ is a finite set of kinded symbols called \emph{variables},
    $T$ is a finite set of kinded symbols called \emph{terminals}, and 
    $N$ is a finite set of kinded symbols called \emph{non-terminals} with
    $S \in N$ the \emph{initial symbol}.
    The sets $V$, $T$, and $N$ are pairwise disjoint.
    The finite set $R$ consists of \emph{rewriting rules} of the form $F = \lambda x_1 \ldots \lambda x_n.e$, where $F \in N$ is a non-terminal of kind $\kappa_1 \to \ldots \kappa_n \to \ground$,
    $x_1, \ldots, x_n \in V$ are variables of the required kinds,
    and $e$ is a $\lambda$-free, variable-closed term of ground kind from $\calT^\ground (T \dotcup N \dotcup \set{ x_1 \colon \kappa_1, \ldots, x_n \colon \kappa_n} )$.
\end{definition}
The semantics of $G$ is defined by rewriting subterms according to the rules in $R$.
A \emph{context} is a term $C[\contextplaceholder] \in \calT (\Gamma \dotcup \set{ \contextplaceholder \colon \ground })$ in which $\contextplaceholder$ occurs exactly once.
Given a context $C[\contextplaceholder]$ and a term $t:\ground$, we obtain $C[t]$ by replacing the unique occurrence of $\contextplaceholder$ in $C[\contextplaceholder]$ by $t$.
With this, $t \derive_G t'$ if there is a context $C[\contextplaceholder]$, a rule $F = \lambda x_1 \ldots \lambda x_n.e$, and a term $F\ t_1\ \ldots\ t_n:\ground$ such that $t = C[F\ t_1\ \ldots\ t_n]$ and
$t' = C\left[e[x_1 \mapsto t_1, \ldots, x_n \mapsto t_n]\right]$.
In other words, we replace one occurrence of $F$ in $t$ by a right-hand side of a rewriting rule, while properly instantiating the variables.
We call such a replaceable $F\ t_1\ \ldots\ t_n$ a \emph{reducible expression (redex)}.
The rewriting step is \emph{outermost to innermost (OI)} if there is no redex that contains the rewritten one as a proper subterm.
The  OI-language $\lang{G}$ of $G$ is the set of all (finite, ranked, labeled) trees $T$ over the terminal symbols that can be created from the initial symbol $S$ via OI-rewriting steps.
We will restrict the rewriting relation to OI-rewritings in the rest of this paper.
Note, all words derivable by IO-rewriting are also derivable with OI-rewriting.

\subparagraph*{Word-Generating Schemes.}

We consider \emph{word-generating schemes}, \ie schemes with terminals $T\disunion \set{\wordend:\ground}$ where exactly one terminal symbol $\wordend$ has kind $\ground$ and all others are of kind $\ground\to \ground$.
The generated trees have the shape
$a_1\ (a_2\ (\cdots\ (a_k\ \wordend)))$,
which we understand as the finite word $a_1 a_2 \ldots a_k \in T^*$.
We also see $\lang{G}$ as a language of finite words.

\subparagraph*{Determinism.}
The above schemes are non-deterministic in that several rules may rewrite a non-terminal.
We associate with a non-deterministic scheme $G=(V, N, T, R, S)$ a deterministic scheme $\detof{G}$ with exactly one rule per non-terminal.
Intuitively, $\detof{G}$ makes the non-determinism explicit with new terminal symbols.

Formally, let $F:\kappa$ be a non-terminal with rules $F= t_1$ to $F= t_{\ell}$.
We may assume each $t_i = \lambda x_{1} \ldots \lambda x_{k}. e_i$, where $e_i$ is $\lambda$-free.
We introduce a new terminal symbol
$\opof{F}: \ground \to \ground \to \ldots \to \ground $ of arity~$\ell$.
Let the set of all these terminals be $\detof{T}=\Set{\opof{F}}{F\in N}$.
The set of rules $\detof{R}$ now consists of a single rule for each non-terminal, namely
$F= \lambda x_{1} \ldots \lambda x_{k}. \opof{F}\ e_1\ \cdots\ e_{\ell}$.
The original rules in $R$ are removed.
This yields $\detof{G}=(V, N, T\disunion \detof{T}, \detof{R}, S)$.
The advantage of resolving the non-determinism explicitly is that we can give a semantics to non-deterministic choices that depends on the non-terminal instead of having to treat non-determinism uniformly.

\subparagraph*{Semantics.}

Let $G=(V, N, T, R, S)$ be a deterministic scheme.
A \emph{model} of $G$ is a pair $\model=(\domain, \inter)$, where $\domain$ is a family of domains $(\domainof{\kappa})_{\kappa\in\Kappa}$ that satisfies the following:
$\domainof{\ground}$ is a cppo and
$\domainof{\kappa_1\rightarrow \kappa_2}=\contfun{\domainof{\kappa_1}}{\domainof{\kappa_2}}$.
Here, $\contfun{A}{B}$ is the set of all \mbox{$\meet$-continuous} functions from domain $A$ to $B$.
We comment on this cppo in a moment.
The interpretation $\inter: T\rightarrow \domain$ assigns to each terminal $s:\kappa$ an element $\interof{s}\in\domainof{\kappa}$.


The ordering on functions is defined component-wise,
$f \leq_{\kappa_1 \to \kappa_2} g$ if
$(f\ x) \leq_{\kappa_2} (g\ x)$ for all $x\in\domainof{\kappa_1}$.
For each $\kappa$, we denote the top element of $\domainof{\kappa}$ by $\top_\kappa$.
For the ground kind, $\top_\ground$ exists since $\domainof{\kappa}$ is a cppo, and $\top_{\kappa_1 \to \kappa_2}$ is the function that maps every argument to $\top_{\kappa_2}$.
The meet of a descending chain of functions $(f_i)_{i \in \N}$ is the function defined by
$(\bigsqcap_{\kappa_1 \to \kappa_2} (f_i)_{i \in \N})\ x =
 \bigsqcap_{\kappa_2} (f_i\ x)_{i \in \N}$.
Note that the sequence on the right-hand side is a descending chain.

The \emph{semantics of terms} defined by a model is a function
\begin{align*}
    \semunder{\model}{-}:\calT\rightarrow (N\disunion V\prightarrow \domain)\rightarrow\domain \ .
\end{align*}
that assigns to each term built over the non-terminals and terminals again a function.
This function expects a valuation $\nu:N\disunion V\prightarrow \domain$ and returns an element from the domain.
A valuation is a partial function that is defined on all non-terminals and the free variables.
We lift $\meet$ to descending chains of valuations with
$(\bigsqcap_{i \in \N} \nu_i)(y) = \bigsqcap_{i \in \N} (\nu_i(y))$
for
$y \in N \disunion V$.
We obtain that the set of such valuations is a cppo where the greatest  elements are those valuations which assign the greatest elements of the appropriate domain to all arguments.

Since the right-hand sides of the rules in the scheme are variable-closed, we do not need a variable valuation for them.
We need the variable valuation, however, whenever we proceed by induction on the structure of terms.
The semantics is defined by such an induction:
\heveaornot{
    \[
        \begin{array}{rcl}
            \semunder{\model}{s}\ \nu  &=& \interof{s} \\
            \semunder{\model}{F}\ \nu  &=& \nu(F) \\
            \semunder{\model}{t_1\ t_2}\ \nu &=& (\semunder{\model}{t_1}\ \nu)\ (\semunder{\model}{t_2}\ \nu) \\
            \semunder{\model}{x}\ \nu &=& \nu(x) \\
            \semunder{\model}{\lambda x:\kappa.t_1}\ \nu &=& d\in \domainof{\kappa} \mapsto \semunder{\model}{t_1}\ \nu[x\mapsto d]
            \ .
        \end{array}
    \]
}{
    \begin{align*}
        \semunder{\model}{s}\ \nu
        & = \interof{s}
        &
        \semunder{\model}{F}\ \nu
        & = \nu(F)
        &
        \semunder{\model}{t_1\ t_2}\ \nu
        & = (\semunder{\model}{t_1}\ \nu)\ (\semunder{\model}{t_2}\ \nu)
        \\
        \semunder{\model}{x}\ \nu
        & = \nu(x)
        &&&
        \semunder{\model}{\lambda x:\kappa.t_1}\ \nu
        & = d\in \domainof{\kappa} \mapsto \semunder{\model}{t_1}\ \nu[x\mapsto d]
        \ .
    \end{align*}
}
We show that
$\semunder{\model}{t}$
is $\meet$-continuous for all terms $t$.
This follows from continuity of the functions in the domain, but requires some care when handling application.

\begin{proposition}
\label{prop:sem-cont}
    For all $t$, $\semunder{\model}{t}$ is $\meet$-continuous (in
    $\nu$) over the respective lattice.
\end{proposition}
Given $\model$, the rules $F_1=t_1,\ldots, F_k=t_k$ of the (deterministic) scheme give a function
\[
    \rhsfun{\model} : (N \rightarrow \domain) \rightarrow (N \rightarrow \domain)\ ,
    \quad\text{ where}\quad
    \rhsfunof{\model}{\nu}(F_j) = \semunder{\model}{t_j}\ \nu \ .
\]
Since the right-hand sides are variable-closed, the $\semunder{\model}{t_j}$ are functions in the non-terminals.
Provided $\semunder{\model}{t_1}$ to $\semunder{\model}{t_k}$ are $\meet$-continuous (in the valuation of the non-terminals), the function $\rhsfun{\model}$ will be $\meet$-continuous.
This allows us to apply Kleene iteration as follows.
The initial value is the greatest element $\soliter{\model}{0}$ where
$\soliter{\model}{0}(F_j) = \top_j$
with $\top_j$ the top element of $\domainof{\kappa_j}$.
The $\nth{(i+1)}$ approximant is computed by evaluating the right-hand side at the  $\nth{i}$ solution,
$\soliter{\model}{i+1} = \rhsfunof{\model}{\sol{\model}^{i}}$.
The greatest fixed point is the tuple $\sol{\model}$ defined below.
It can be understood as the greatest solution to the equation $\nu = \rhsfunof{\model}{\nu}$.
We call this greatest solution $\sol{\model}$ the \emph{semantics of the scheme} in the model.
\[
    \sol{\model} = \bigsqcap_{i \in \N} {\soliter{\model}{i}} = \bigsqcap_{i \in \N} \rhsfun{\model}^i (\soliter{\model}{0})
\]
%


\section{Higher-Order Inclusion Games}
\label{Section:Games}

Our goal is to solve higher-order games, whose arena is defined by a scheme.
We assume that the derivation process is controlled by two players.
To this end, we divide the non-terminals of a word-generating scheme into those owned by the existential player $\playere$ and those owned by the universal player $\playera$.
Whenever a non-terminal is to be replaced during the derivation, it is the owner who chooses which rule to apply.
The winning condition is given by an automaton~$A$, Player $\playere$ attempts to produce a word that is in $\lang{A}$, while Player $\playera$ attempts to produce a word outside of $\lang{A}$.

\begin{definition}
    A \emph{higher-order game} is a triple
    $\horsgame = (G, A, O)$
    where
    $G$ is a word-generating scheme,
    $A$ is an NFA,
    $O : N \rightarrow \{\playere, \playera\}$ is a partitioning of the non-terminals of $G$.
\end{definition}
A play of the game is a sequence of OI-rewriting steps.
Since terms generate words, it is unambiguous which term forms the next redex to be rewritten.
In particular, all terms are of the form
$a_1 ( a_2 ( \cdots (a_k (t) ) ) )$,
where $t$ is either $\wordend$ or a redex
$F\ t_1\ \cdots\ t_m$.
If
$O(F) = \playere$
then Player $\playere$ chooses a rule
$F = \lambda x_1 \ldots \lambda x_m . e$
to apply, else Player $\playera$ chooses the rule.
This moves the play to
$a_1\ (a_2\ (\cdots\ (a_k\ e[x_1 \mapsto t_1, \ldots, x_m \mapsto t_m])))$.

Each play begins at the initial non-terminal $S$, and continues either ad infinitum or until a term
$a_1\ (a_2\ (\cdots\ (a_k\ \wordend)))$, understood as the word $w=a_1\ldots a_k$, is produced.
Infinite plays do not produce a word and are won by Player $\playere$.
Finite maximal plays produce such a word $w$.
Player $\playere$ wins whenever $w \in \lang{A}$, Player $\playera$ wins if $w \in \overline{\lang{A}}$.
Since the winning condition is Borel, either Player $\playere$ or Player $\playera$ has a winning strategy~\cite{M75}.

{
\sffamily
\begin{problem}
    \problemtitle{The Winner of a Higher-Order Game}
    \problemshort{($\gamehorsreg$)}
    \probleminput{A higher-order game $\horsgame$.}
    \problemquestion{
        Does Player $\playere$ win $\horsgame$?
        If so, effectively represent Player $\playere$'s strategy.}
\end{problem}
}

\noindent
Our contribution is a fixed-point algorithm to decide $\gamehorsreg$.
We derive it in three steps.
First, we develop a concrete model for higher-order games whose semantics captures the above winning condition.
Second, we introduce a framework that for two models and a mapping between them guarantees that the mapping of the greatest fixed point with respect to the one model is the greatest fixed point with respect to the other model. 
Finally, we introduce an abstract model that uses a finite ground domain.
The solution of $\gamehorsreg$ can be read off from the semantics in the abstract model, which in turn can be computed via Kleene iteration.
Moreover, this semantics can be used to define Player~$\playere$'s winning strategy.
We instantiate the framework for the concrete and abstract model to prove the soundness of the algorithm.

\subsection*{Concrete Semantics}

Consider a $\gamehorsreg$ instance $\horsgame = (G, A, O)$.
Let $\detof{G}$ be the determinized version of $G$.
Our goal is to define a model $\modelconc = (\domainconc, \interconc)$ such that the semantics of $\detof{G}$ in this model allows us to decide $\gamehorsreg$.
Recall that we only have to define the ground domain.
For composed kinds, we use the functional lifting discussed in Section~\ref{Section:HORS}.

Our idea is to associate to kind $\ground$ the set of positive Boolean formulas where the atomic propositions are words in $T^\ast$.
To be able to reuse the definition, we define formula domains in more generality as follows.

\subparagraph{Domains of Boolean Formulas}

Given a (potentially infinite) set $P$ of atomic propositions, the \emph{positive Boolean formulas} $\boolof{P}$ over $P$ are defined to contain $\ftrue$, every $p$ from $P$, and compositions of formulas via conjunction and disjunction.
We work up to logical equivalence, which means we treat $\fmla_1$ and $\fmla_2$ as equal as long as they are logically equivalent.

Unfortunately, if the set $P$ is infinite, $\boolof{P}$ is not a cppo, because the meet of a descending chain of formulas might not be a finite formula.
The idea of our domain is to have conjunctions of infinitely many formulas.
As is common in logic, we represent them as infinite sets.
Therefore, we consider the set of all sets of (finite) positive Boolean formulas
$\pwrset{\boolof{T^*}}\setminus\set{\emptyset}$
factorized modulo logical equivalence, denoted
$(\pwrset{\boolof{T^*}}\setminus\set{\emptyset}) \slash_\Leftrightarrow$.
To be precise, the sets may be finite or infinite, but they must be non-empty.

To define the factorization, let an assignment to the atomic propositions be given by a subset of $P' \subseteq P$.
The atomic proposition $p$ is true if $p \in P'$.
An assignment satisfies a Boolean formula, if the formula evaluates to true in that assignment.
It satisfies a set of Boolean formulas, if it satisfies all elements.
Given two sets of formulas $\setfmla_1$ and $\setfmla_2$, we write $\setfmla_1\Rightarrow \setfmla_2$, if every assignment that satisfies $\setfmla_1$ also satisfies $\setfmla_2$.
Two sets of formulas are equivalent, denoted $\setfmla_1\Leftrightarrow \setfmla_2$, if $\setfmla_1\Rightarrow \setfmla_2$ and $\setfmla_2\Rightarrow \setfmla_1$ holds.

The ordering on these factorized sets is implication (which by transitivity is independent of the representative).
The top element is the set $\set{\ftrue}$, which is implied by every set.
The conjunction of two sets is union.
Note that it forms the meet in the partial order, and moreover note that meets over arbitrary sets exist, in particular the domain is a cppo.
We will also need an operation of disjunction, which is defined by
\mbox{$
    \setfmla_1\vee \setfmla_2 = \Set{\fmla_1\vee \fmla_2}{\fmla_1\in \setfmla_1, \fmla_2\in \setfmla_2}.
$}
We will also use disjunctions of higher (but finite) arity where convenient.
Note that the disjunction on finite formulas is guaranteed to result in a finite formula.
Therefore, the above is well-defined.

In our case, the assignment $P' \subseteq T^*$ of interest is the language of the automaton $A$.
Player $\playere$ will win the game iff the concrete semantics assigns a set of formulas to $S$ that is satisfied by $\lang{A}$.

\subparagraph{The Concrete Domains and Interpretation of Terminals.}

From a ground domain, higher-order domains are defined as continuous functions as in Section~\ref{Section:HORS}.
Thus we only need
\[
    \domainconcof{\ground} = \left( \pwrset{\boolof{T^*}}\setminus\set{\emptyset} \right) \slash_\Leftrightarrow \ .
\]
The endmarker $\wordend$ yields the set of formulas $\set{\varepsilon}$, \ie $\interconcof{\wordend}=\set{\varepsilon}$.
A terminal $a:\ground\to \ground$ prepends $a$ to a given word $w$.
That is $\interconcof{a} = \prependf{a}$, where $\prependf{a}$ distributes over conjunction and disjunction:
\[
    \prepend{a}{\fmla}
    =
    \left\{
    \begin{array}{ll}
        aw & \fmla = w\ ,
        \\
        \prepend{a}{\fmla_1}
        \genop
        \prepend{a}{\fmla_2}
        &
        \fmla = \fmla_1\genop\fmla_2 \text{ and } \genop \in \set{\land,\lor}\ ,
        \\
        \fmla & \fmla = \ftrue \ .
        \\
    \end{array}
    \right.
\]
We apply $\prependf{a}$ to sets of formulas by applying it to every element.
Finally, $\interconcof{\opof{F}}$
where $\opof{F}$ has arity $\ell$ is an $\ell$-ary conjunction (\resp disjunction) if Player~$\playera$ (\resp $\playere$) owns~$F$.

For $\modelconc = (\domainconc, \interconc)$ to be a model, we need our interpretation of terminals to be $\meet$-continuous.
This follows largely by the distributivity of our definitions.

\begin{lemma}
\label{Lemma:ContinuityRHSConc}
    For all non-ground terminals $s$, $\interconcof{s}$ is $\meet$-continuous.
\end{lemma}
%
%
%
\begin{example}
Consider the higher-order game defined by the scheme $S=H\ a\ \$\ |\ b\ \$$ and $H=\lambda f.\lambda x.f\ (f\ x)\ |\ \lambda f.\lambda x.H\ (H\ f)\ x$.
Assume $S$ is owned by Player~$\playere$ and $H$ is owned by Player~$\playera$.
Let the automaton accept the language $\set{b}$.
Player~$\playere$ can choose to rewrite $S$ to $b\ \$$ and therefore has a strategy to produce a word in the language. 
To derive this information from the concrete semantics, we compute $\sol{\modelconc}(H)$. 
It is the function mapping $f\in \contfun{\domainconcof{\ground}}{\domainconcof{\ground}}$ and $d\in \domainconcof{\ground}$ to $\bigcup_{k>0}f^{2k}(d)$. 
Note that the union is the conjunction of sets of formulas, which is the interpretation of $\opof{H}$ for the universal player. 
Moreover, note that due to non-determinism we obtain all even numbers of applications of~$f$, not only the powers of $2$. 
With this, the semantics of the initial symbol is
\begin{align*}
\sol{\modelconc}(S) = \bigcup_{k>0}\prependf{a}^{2k}(\set{\varepsilon}) \vee \prepend{b}{\set{\varepsilon}}=\Set{a^{2k}\vee b}{k>0}.
\end{align*}
The assignment $\set{b}$ given by the language of the NFA satisfies $\Set{a^{2k}\vee b}{k>0}$. Indeed, since $b$ evaluates to true, every formula in the set evaluates to true.
\end{example}

\subparagraph{Correctness of Semantics and Winning Strategies.}

We need to show that the concrete semantics matches the original semantics of the game.

\begin{theorem} \label{Theorem:ConcGood}
    $\solconcof{S}$ is satisfied by $\lang{A}$ iff there is a winning strategy for Player $\playere$.
\end{theorem}
When $\solconcof{S}$ is satisfied by $\lang{A}$ the concrete semantics gives a winning strategy for~$\playere$: 
From a term $t$ such that $\semunder{\modelconc}{t}\ \solconc$ is satisfied by $\lang{A}$, Player $\playere$, when able to choose, picks a rewrite rule that transforms $t$ to $t'$, where $\semunder{\modelconc}{t'}\ \solconc$ remains satisfied.
The proof of Theorem~\ref{Theorem:ConcGood} shows this is always possible, and, moreover, Player $\playera$ is unable to reach a term for which satisfaction does not hold.
This does not yet give an effective strategy since we cannot compute $\semunder{\modelconc}{t}\ \solconc$.
However, the abstract semantics will be computable, and can be used in place of the concrete semantics by Player $\playere$ to implement the winning strategy.

The proof that $\solconcof{S}$ being unsatisfied implies a winning strategy for Player~$\playera$ is more involved and requires the definition of a correctness relation between semantics and terms that is lifted to the level of functions, and shown to hold inductively.


\section{Framework for Exact Fixed-Point Transfer}
\label{Section:Framework}

The concrete model $\modelconc$ does not lead to an algorithm for solving $\gamehorsreg$ since its domains are infinite.
Here, we consider an abstract model $\modelabs$ with finite domains.
The soundness of the resulting Kleene iteration relies on the two semantics being related by a precise abstraction~$\abs$.
Since both semantics are defined by fixed points, this requires us to prove $\abs (\solconc) = \solabs$.
In this section, we provide a general framework to this end.

Consider the deterministic scheme $G$ together with two models (left and right) $\modell~=~(\domainl, \interl)$ and
$\modelr = (\domainr, \interr)$.
Our goal is to relate the semantics in these models in the sense that $\sol{\modelr}=\absof{\sol{\modell}}$.
Such exact fixed-point transfer results are well-known in abstract interpretation.
To generalize them to higher-order we give easy to instantiate conditions on $\abs$, $\modell$, and $\modelr$ that yield the above equality.
Interestingly, exact fixed-point transfer results seem to be rare for higher-order (e.g.~\cite{R13}).
Our development is inspired by Abramsky's lifting of abstraction functions to logical relations~\cite{Abramsky90}, which generalizes~\cite{BHA86,AH87}.
These works focus on approximation and the compatibility we need for exactness is missing.
Our framework is easier to apply than~\cite{CousotCousot94,BackhouseBackhouse04}, which are again concerned with approximation and do not offer (but may lead to) exact fixed-point transfer results.

For the terminology, an \emph{abstraction} is a function $\abs:\domainlof{\ground}\rightarrow \domainrof{\ground}$.
To lift the abstraction to function domains, we define the notion of \emph{being compatible with $\abs$}.
Compatibility intuitively states that the function on the concrete domain is not more precise than what the abstraction function distinguishes.
This allows us to define the abstraction of a function by applying the function and abstracting the result, $\absof{f}\ \absof{\vall}=\absof{f\ \vall}$.
Compatibility ensures the independence of the choice of $\vall$.

By definition, all ground elements $\vall\in\domainlof{\ground}$ are compatible with $\abs$.
For function domains, compatibility and the abstraction are defined as follows.

\begin{definition}
    Assume $\abs$ and the notion of compatibility are defined on $\domainlof{\kappa_1}$ and $\domainlof{\kappa_2}$.
    Let $\toplof{\kappa}$ (resp.\ $\toprof{\kappa}$) be the greatest element of $\domainlof{\kappa}$ (resp. $\domainrof{\kappa}$) for each $\kappa$.

    \begin{compactenum}
        \item
            Function $f\in \domainlof{\kappa_1\rightarrow \kappa_2}$ is compatible with $\abs$, if
            \begin{compactenum}
                \item
                    for all compatible
                    $\vall, \vall'\in \domainlof{\kappa_1}$
                    with
                    $\absof{\vall}=\absof{\vall'}$
                    we have
                    $\absof{f\ \vall}=\absof{f\ \vall'}$, and

                \item
                    for all compatible
                    $\vall \in \domainlof{\kappa_1}$
                    we have that
                    $f\ \vall$
                    is compatible.
            \end{compactenum}
        \item
            We define
            $\absof{f} \in \domainrof{\kappa_1\rightarrow \kappa_2}$ as follows.
            \begin{compactenum}
                \item
                    If $f$ is compatible,
                    we set $\absof{f}\ \valr= \absof{f\ \vall}$, provided there is a compatible $\vall\in \domainlof{\kappa_1}$ with $\valr=\absof{\vall}$, and $\absof{f}\ \valr=\toprof{\kappa_2}$ otherwise.
                \item
                    If $f$ is not compatible, $\absof{f}=\toprof{\kappa_1\rightarrow \kappa_2}$.
            \end{compactenum}
    \end{compactenum}
%
    We lift $\abs$ to valuations
    $\nu : N \disunion V \prightarrow \domainl$
    by
    $\absof{\nu}(F) = \absof{\nu(F)}$
    and similar for $x$.
    We also lift compatibility to valuations
    $\nu:N\disunion V\prightarrow \domainl$
    by requiring $\nu(F)$ to be compatible for all
    $F\in N$ and similar for $x\in V$.
\end{definition}
%
%
%
The conditions needed for the exact fixed-point transfer are the following.

\begin{definition}
    Function $\abs$ is \emph{precise} for $\modell$ and $\modelr$, if
    \begin{compactitem}
    \item[$\precisionsur$] $\absof{\domainlof{\ground}}=\domainrof{\ground}$,
    \item[$\precisioncont$] $\abs:\domainlof{\ground}\rightarrow \domainrof{\ground}$ is $\meet$-continuous,
    \item[$\precisionbot$] $\absof{\toplof{\ground}}=\toprof{\ground}$,
    \item[$\precisionrel$]
            $\absof{\interlof{s}}= \interrof{s}$ for all terminals $s\colon\ground$,
            and similarly $\absof{\interlof{s}\ \vall}= \interrof{s}\ \absof{\vall}$ for all terminals $s:\kappa_1 \rightarrow \kappa_2$ and all compatible $\vall\in\domainlof{\kappa_1}$,
    \item[$\precisioncomp$] $\interlof{s}\ \vall$ is compatible for all terminals $s:\kappa_1 \rightarrow \kappa_2$, and all compatible $\vall\in\domainlof{\kappa_1}$.
    \end{compactitem}
\end{definition}

\noindent
$\precisionsur$ is surjectivity of $\abs$.
$\precisioncont$ states that $\abs$ is well-behaved wrt.\ $\meet$.
$\precisionbot$ says that the greatest element is mapped as expected.
Note that $\precisionsur$-$\precisionbot$ are only posed for the ground domain.
One can prove that they generalize to function domains by the definition of function abstraction.
$\precisionrel$ is that the interpretations of terminals in $\modelconc$ and $\modelabs$ are suitably related.
Finally~$\precisioncomp$ is compatibility.
$\precisionrel$ and $\precisioncomp$ are generalized to terms in Lemma~\ref{Lemma:MainTransfer}.

To prove $\absof{\sol{\modell}}=\sol{\modelr}$, we need that $\rhsfunr$ is an exact abstract transformer of $\rhsfunl$.
The following lemma states this for all terms $t$, in particular those that occur in the equations.
The generalization to product domains is immediate.
Note that the result is limited to compatible valuations, but this will be sufficient for our purposes.
The proof proceeds by induction on the structure of terms, while simultaneously proving $\semunder{\modell}{t}$ compatible with $\abs$.
With this result, we obtain the required exact fixed-point transfer for precise abstractions.

\begin{lemma}\label{Lemma:MainTransfer}
    Assume $\precisionsur$, $\precisionrel$, and $\precisioncomp$ hold.
    For all terms $t$ and all compatible $\nu$, we have
    $\semunder{\modell}{t}\ \nu$ compatible and
    $\absof{\semunder{\modell}{t}\ \nu}=\semunder{\modelr}{t}\ \absof{\nu}$.
\end{lemma}

\begin{theorem}[Exact Fixed-Point Transfer]
\label{Theorem:FixedPointTransfer}
    Let $G$ be a scheme with models $\modell$ and $\modelr$.
    Let $\solfl$ and $\solfr$ be the corresponding semantics.
    If $\abs: \domainl \rightarrow \domainr$ is precise, we have $\solfr=\absof{\solfl}$.
\end{theorem}

\section{Domains for Higher-Order Games}
\label{Section:Instantiation}

We propose two domains, \emph{abstract} and \emph{optimized}, that allow us to solve $\gamehorsreg$.
The computation is a standard fixed-point iteration, and, in the optimized domain, this iteration has optimal complexity.
Correctness follows by instantiating the previous framework.

\subparagraph*{Abstract Semantics.}

Our goal is to define an abstract model for games that
(1) suitably relates to the concrete model from Section~\ref{Section:Games} and
(2) is computable.
By a suitable relation, we mean the two models should relate via an abstraction function.
Provided the conditions on precision hold, correctness of the abstraction then follows from Theorem~\ref{Theorem:FixedPointTransfer}.
Combined with Theorem~\ref{Theorem:ConcGood}, this will allow us to solve $\gamehorsreg$.
Computable in particular means the domain should be finite and the operations should be efficiently computable.

We define the $\modelabs = (\domain^\abstag, \inter^\abstag)$ as follows.
Again, we resolve the non-determinism into Boolean formulas.
But rather than tracking the precise words generated by the scheme, we only track the current set of states of the automaton.
To achieve the surjectivity required by precision, we restrict the powerset to those sets of states from which a word is accepted.
Let $\acceptfrom{w} = \Set{ q }{ q \tow{w} q_f \in Q_f }$.
For a language $L$ we have
$\acceptfrom{L} = \Set{ \acceptfrom{w} }{ w \in L }$.
The abstract domain for terms of ground kind is $\domainabsof{\ground} = \boolof{\acceptfrom{T^*}}$.
The lifting to functions is as explained in Section~\ref{Section:HORS}.
Satisfaction is now defined relative to a set $\solset$ of elements of $\pwrset{\QNFA}$ (cf. Section~\ref{Section:Games}).
With finitely many atomic propositions, there are only finitely many formulas (up to logical equivalence).
This means we no longer need sets of formulas to represent infinite conjunctions, but can work with plain formulas.
The ordering is thus the ordinary implication with the meet being conjunction and top being $\ftrue$.

The interpretation of ground terms is $\interabsof{\wordend} = Q_f$ and $\interabsof{a}= \predecf{a}$. Here $\predecf{a}$ is the predecessor computation under label $a$, $\predecf{a}(Q)=\Set{q'\in \QNFA}{q'\tow{a} q \in Q}$.
It is lifted to formulas by distributing it over conjunction and disjunction.
The composition operators are again interpreted as conjunctions and disjunctions, depending on the owner of the non-terminal. 
Since we restrict the atomic propositions to $\acceptfrom{T^*}$, we have to show that the interpretations use only this restricted set.
Proving $\interabsof{s}$ is $\meet$-continuous is standard.

\begin{lemma}
\label{lem:abs-defined}
    The interpretations are defined on the abstract domain.
\end{lemma}

\begin{lemma}
    \label{lem:abs-continuous}
    For all terminals $s$, $\interabsof{s}$ is $\meet$-continuous over the respective lattices.
\end{lemma}
Recall our concrete model is $\modelconc=(\domainconc, \interconc)$,  where $\domainconc=\pwrset{\boolof{T^*}}$.
To relate this model to $\modelabs$, we define the abstraction function $\abs:\domainconcof{\ground}\rightarrow\domainabsof{\ground}$.
It leaves the Boolean structure of a formula unchanged but maps every word (which is an atomic proposition) to the set of states from which this word is accepted. 
For a set of formulas, we take the conjunction of the abstraction of the elements. 
This conjunction is finite as we work over a finite domain, so there is no need to worry about infinite syntax.
Technically, we define $\abs$ on $\boolof{T^*}$ by
$\absof{\setfmla}=\bigwedge_{\fmla\in\setfmla}\absof{\fmla}$
for a set of formulas $\setfmla\in\pwrset{\boolof{T^*}}$, and
\[
    \absof{\fmla}
    =
    \left\{
    \begin{array}{ll}
        \acceptfrom{w} & \text{if } \fmla = w,
        \\
        \absof{\fmla_1}\genop\absof{\fmla_2} &\text{if }
        \fmla = \fmla_1\genop\fmla_2 \text{ and } \genop \in \set{\land,\lor},
        \\
        \fmla &\text{if } \fmla = \ftrue \ .\\
    \end{array}
    \right.
\]
This definition is suitable in that $\absof{\sol{\modelconc}}=\sol{\modelabs}$ entails the following.
\begin{theorem}
\label{Theorem:GameAbstract}
    $\sol{\modelabs}(S)$ is satisfied by
    $\Set{Q \in \acceptfrom{T^*}}{q_0 \in Q}$
    iff Player $\playere$ wins $\horsgame$.
\end{theorem}
To see that the theorem is a consequence of the exact fixed-point transfer, observe that
$\Set{Q \in \acceptfrom{T^*}}{q_0 \in Q}
 =
 \acceptfrom{\lang{A}}$.
Then, by $\sol{\modelabs}=\absof{\sol{\modelconc}}$ we have
$\acceptfrom{\lang{A}}$
satisfies
$\sol{\modelabs}(S)$
iff it also satisfies
$\absof{\sol{\modelconc}(S)}$.
This holds iff $\lang{A}$ satisfies $\sol{\modelconc}(S)$ (a simple induction over formulas).
By Theorem~\ref{Theorem:ConcGood},
this occurs iff Player $\playere$ wins the game.

It remains to establish $\absof{\sol{\modelconc}}=\sol{\modelabs}$.
With the framework, the exact fixed-point transfer follows from precision, Theorem~\ref{Theorem:FixedPointTransfer}.
The proof of the following is routine.

\begin{proposition}
\label{Proposition:PrecisionAbs}
    $\abs$ is precise.
    Hence, $\absof{\sol{\modelconc}}=\sol{\modelabs}$.
\end{proposition}

\subparagraph*{Optimized Semantics.}
The above model yields a decision procedure for $\gamehorsreg$ via Kleene iteration.
Unfortunately, the complexity is one exponential too high:
The height of the domain for a symbol of order $k$ in the abstract model is $(k+2)$-times exponential,
where the height is the length of the longest strictly descending chain in the domain.
This gives the maximum number of steps of Kleene iteration needed to reach the fixed point.

We present an optimized version of our model that is able to close the gap:
In this model, the domain for an order-$k$ symbol is only $(k+1)$-times exponentially~high.
The idea is to resolve the atomic propositions in $\modelabs$, which are sets of states, into disjunctions among the states.
The reader familiar with inclusion algorithms will find this decomposition surprising.

We first define $\abs:\boolof{\acceptfrom{T^*}}\rightarrow \boolof{\QNFA}$.
The optimized domain will then be based on the image of $\abs$.
This guarantees surjectivity.
For a set of states $Q$, we define $\absof{Q}=\bigvee Q= \bigvee_{q\in Q} q$.
For a formula, the abstraction function is defined to distribute over conjunction and disjunction.
The optimized model is $\modelopt = (\domain^\opttag, \interopt)$ with ground domain $\absof{\boolof{\acceptfrom{T^*}}}$.
The interpretation is $\interoptof{\wordend}=\bigvee Q_f$.
For $a$, we resolve the set of predecessors into a disjunction,
$\interoptof{a}\ q = \bigvee \predecf{a}(\set{q})$.
The function distributes over conjunction and disjunction.
Finally, $\interoptof{\opof{F}}$ is conjunction or disjunction of formulas, depending on the owner of the non-terminal.
Since we use a restricted domain, we have to argue that the operations do not leave the domain.
It is also straightforward to prove our interpretation is $\meet$-continuous as required.

\begin{lemma}\label{Lemma:OptDefined}
    The interpretations are defined on the optimized domain.
\end{lemma}

\begin{lemma}
\label{Lemma:OptContinuous}
    For all terminals $s$, $\interoptof{s}$ is $\meet$-continuous over the respective lattices.
\end{lemma}
We again show precision, enabling the required exact fixed-point transfer.
\begin{proposition}
\label{Proposition:PrecisionOpt}
    $\abs$ is precise.
    Hence, $\absof{\sol{\modelabs}}=\sol{\modelopt}$.
\end{proposition}

\begin{theorem}\label{Theorem:GameOptimized}
$\sol{\modelopt}(S)$ is satisfied by $\set{q_0}$ iff Player~$\playere$ wins $\horsgame$.
\end{theorem}
It is sufficient to show $\sol{\modelabs}(S)$ is satisfied by $\Set{Q \in \acceptfrom{T^*}}{q_0 \in Q}$ iff $\sol{\modelopt}(S)$ is satisfied by $\set{q_0}$.
Theorem~\ref{Theorem:GameAbstract} then yields the statement. 
Propositions $Q$ in $\sol{\modelabs}(S)$ are resolved into disjunctions $\bigvee Q$ in $\sol{\modelopt}(S)$.
For such a proposition, we have $Q\in \Set{Q \in \acceptfrom{T^*}}{q_0 \in Q}$ iff $\bigvee Q$ is satisfied by $\set{q_0}$.
This equivalence propagates to the formulas $\sol{\modelabs}(S)$ and $\sol{\modelopt}(S)$ as the Boolean structure coincides.
The latter follows from $\absof{\sol{\modelabs}(S)}=\sol{\modelopt}(S)$.

\subparagraph*{Complexity.}

To solve $\gamehorsreg$, we compute the semantics $\sol{\modelopt}$ and then evaluate $\sol{\modelopt}(S)$ at the assignment $\set{q_0}$.
For the complexity, assume that the highest order of any non-terminal in $\horsgame$ is $k$.
We show the number of iterations needed to compute the greatest fixed point is at most  $(k+1)$-times exponential.
We do this via a suitable upper bound on the length of strictly descending chains in the domains assigned by $\domainopt$.

\begin{proposition}
\label{lem:membership}
    The semantics $\sol{\modelopt}$ can be computed in $(k+1)\EXPTIME$, where $k$ is the highest order of any non-terminal in the input scheme.
\end{proposition}
The lower bound is via a reduction from the word membership problem for alternating $k$-iterated pushdown automata with polynomially-bounded auxiliary work-tape.
This problem was shown by Engelfriet to be $(k+1)\mathsf{EXP}$-hard.
We can reduce this problem to $\gamehorsreg$ via well-known translations between iterated stack automata and recursion schemes, using the regular language specifying the winning condition to help simulate the work-tape.

\begin{proposition}
\label{lem:hardness}
    Determining whether Player $\playere$ wins $\horsgame$ is $(k+1)\mathsf{EXP}$-hard for $k > 0$.
\end{proposition}
Together, these results show the following corollary and final result.

\begin{corollary}
    \label{Theorem:Complexity}
    $\gamehorsreg$ is $(k+1)\mathsf{EXP}$-complete for order-$k$ schemes and $k > 0$.
\end{corollary}

%

\newpage
\bibliographystyle{plainurl}
\bibliography{cited}

\newpage
\appendix

\section{Relation to Higher-Order Model Checing}
\label{Appendix:Introduction}

We elaborate on the relation of our work to the influential line of research on intersection types as pioneered by~\cite{KO09}.
With intersection types, it is usually proven that \emph{there is} a word or tree derivable by a HORS that is accepted by an automaton, i.e.\ a well-typed type environment can be certificate for the non-emptiness of the intersection $\lang{\mathit{scheme}} \cap \lang{\mathit{Automaton}} \neq \emptyset$.
If the HORS is deterministic, $\lang{\mathit{scheme}}$ consists of a single tree, so this is also decides the inclusion $\lang{\mathit{scheme}} \subseteq \lang{\mathit{Automaton}}$.
If we naively extend intersection types to non-deterministic schemes, this is not true anymore.
To prove the inclusion in this case, we will need to complement the automaton and prove the emptiness of the intersection, \ie $\lang{\mathit{scheme}} \cap \lang{\overline{\mathit{Automaton}}} = \emptyset$.
Note that a well-typing (a well-typed type environment) cannot prove the emptiness by itself: If the type for the initial symbol does not contain a transition from a final to an initial state, that can either stem from the non-existence of a such a transition sequence, or from the typing not being strong enough.
For example, the empty typing that does not assign any type to any symbol (or the empty intersection, if you want), is a well-typing and does not prove anything.
Therefore, an algorithm that decides the non-emptiness of the intersection by using intersection-types has to guarantee that it constructs a well-typing strong enough to prove the existence of an accepting transition sequence if such a sequence exists.
Note that algorithms that compute intersection types usually allow alternating automata as the specification.
It is conceptually easier to complement an alternating automaton than it is to complement a non-deterministic automaton: The transition for each origin and label is given as a Boolean formula, and we can get the complement automaton by considering the dual formula (\ie the formula in which conjunctions and disjunctions are swapped).
Note that usually, the transition formulas are normalized to disjunctive normal form (DNF), so computing the dual formula (which will then be in CNF) and re-normalizing it to DNF can lead to an exponential blowup.

Work by Neatherway~\textit{et al.}~\cite{NRO12} and Ramsay~\cite{R13} considers schemes with non-determinism in the form of case statements.
To handle this non-determinism they introduce union types as a ground type.
Neatherway~\textit{et al.} give an optimised algorithm for checking such schemes against deterministic trivial automata (where all infinite runs are accepting -- i.e.\ a B\"uchi condition where all states are accepting).
In his thesis, Ramsay extends this to checking non-deterministic schemes against non-deterministic trivial automata using abstract interpretation from schemes to types.
In our work, we generalise non-determinism to games (played over word-generating schemes), with a non-deterministic target language.


\section{Proofs for Section~\ref{Section:HORS}}

\subsection{Proof of Proposition~\ref{prop:sem-cont}}
\label{Section:sem-cont-proof}

\begin{proof}
    Let $\left(\nu_i\right)_{i \in \N}$ be a descending chain of evaluations, \ie $\nu_i \geq \nu_{i+1}$ for all $i \in \N$.
    It is to show that for all $t$, $\semunder{\model}{t}$ is $\meet$-continuous (in the argument $\nu$) over the respective lattice, \ie
    \[
        \semunder{\model}{t} \left( \bigsqcap_{i \in \N} \nu_i \right)
        =
        \bigsqcap\limits_{i \in \N} (\semunder{\model}{F}\ \nu_i)
        \ .
    \]
    We proceed by induction over $t$.
    \begin{enumerate}
        \item
            \textsfbf{Case $t = F$ or $t = x$.}\\
            Both of these cases are identical, hence we only show the former.
            We have
            \[
                \semunder{\model}{F}\ (\bigsqcap\limits_{i \in \N} \nu_i) =
                (\bigsqcap\limits_{i \in \N} \nu_i)(F) =
                \bigsqcap\limits_{i \in \N} (\nu_i(F)) =
                \bigsqcap\limits_{i \in \N} (\semunder{\model}{F}\ \nu_i)
            \]
            where the first and final equalities are by definition of the concrete semantics, and the second is by definition of $\meet$ over valuations $\nu_i$.
        \item
            \textsfbf{Case $t = s$ for some terminal $s$.}\\
            Similar to the previous case, we have
            \[
                \semunder{\model}{s}\ (\bigsqcap\limits_{i \in \N} \nu_i) =
                \interof{s} =
                \bigsqcap\limits_{i \in \N} (\semunder{\model}{s}\ \nu_i)
            \]
            by definition.

        \item
            \textsfbf{Case $t = t_1\ t_2$.}\\
            We have
            \begin{align*}
                &\ \semunder{\model}{t_1\ t_2}\ (\bigsqcap\limits_{i \in \N} \nu_i)
                \\
                \text{(Definition of semantics)}
                =&\ %
                (\semunder{\model}{t_1}\ (\bigsqcap\limits_{i \in \N} \nu_i))\ %
                (\semunder{\model}{t_2}\ (\bigsqcap\limits_{i \in \N} \nu_i))
                \\
                \text{(Induction hypothesis)}
                =&\ %
                (\bigsqcap\limits_{i \in \N} (\semunder{\model}{t_1}\ \nu_i))\ %
                (\bigsqcap\limits_{i \in \N} (\semunder{\model}{t_2}\ \nu_i))
                \\
                \text{(Definition of $\meet$ for functions)}
                =&\ %
                \bigsqcap\limits_{i \in \N} (
                    (\semunder{\model}{t_1}\ \nu_i)\ %
                    (\bigsqcap\limits_{i \in \N} (\semunder{\model}{t_2}\ \nu_i))
                )
                \\
                \text{(Continuity of $\semunder{\model}{t_1}\ \nu_i \in \domain$)}
                =&\ %
                \bigsqcap\limits_{i \in \N}
                \bigsqcap\limits_{j \in \N} (
                    (\semunder{\model}{t_1}\ \nu_i)\ (\semunder{\model}{t_2}\ \nu_j))
                )
                \\
                \text{(Argued below)}
                =&\ %
                \bigsqcap\limits_{i \in \N} (
                    (\semunder{\model}{t_1}\ \nu_i)\ (\semunder{\model}{t_2}\ \nu_i))
                )
                \\
                \text{(Definition of semantics)}
                =&\ %
                \bigsqcap\limits_{i \in \N} (\semunder{\model}{t_1\ t_2}\ \nu_i) \ .
            \end{align*}
            We have to argue the step indicated above.  That is,
            \[
                \bigsqcap\limits_{i \in \N}
                \bigsqcap\limits_{j \in \N} (
                    (\semunder{\model}{t_1}\ \nu_i)\ (\semunder{\model}{t_2}\ \nu_j))
                )
                =
                \bigsqcap\limits_{i \in \N} (
                    (\semunder{\model}{t_1}\ \nu_i)\ (\semunder{\model}{t_2}\ \nu_i))
                ) \ .
            \]
            The right-hand side is greater than the left-hand side, because terms of the form
            $((\semunder{\model}{t_1}\ \nu_i)\ (\semunder{\model}{t_2}\ \nu_j))$
            where $\nu_i \neq \nu_j$ are missing in the RHS.
            To see that it is in fact equal, note that for two indices $i,j \in \N$, we have either
            $\nu_i \leq \nu_j$
            or
            $\nu_j \leq \nu_i$,
            since the valuations form a descending chain.
            Let $m = \min \set{i,j}$.
            We now use that $\meet$-continuity implies monotonicity, and thus we have
            \[
                ((\semunder{\model}{t_1}\ \nu_m)\ (\semunder{\model}{t_2}\ \nu_m))
                 \leq
                ((\semunder{\model}{t_1}\ \nu_i)\ (\semunder{\model}{t_2}\ \nu_j))
                \ .
            \]
            Hence, for any expression $((\semunder{\model}{t_1}\ \nu_i)\ (\semunder{\model}{t_2}\ \nu_j))$ that is missing in the meet in the RHS, the meet in the RHS contains an expression that is smaller, hence, they are equal.

        \item
            \textsfbf{Case $t = \lambda x . t'$.}\\
            We have
            \begin{align*}
                &\ \semunder{\model}{\lambda x . t'}\ (\bigsqcap\limits_{i \in \N} \nu_i)
                \\
                \text{(Definition of semantics)}
                =&\ %
                v \mapsto (\semunder{\model}{t'}\ (\bigsqcap\limits_{i \in \N} \nu_i)[x \mapsto v])
                \\
                \text{(Induction hypothesis)}
                =&\ %
                v \mapsto (\bigsqcap\limits_{i \in \N} (\semunder{\model}{t'}\ \nu_i[x \mapsto v]))
                \\
                \text{(Definition of $\meet$ for functions)}
                =&\ %
                \bigsqcap\limits_{i \in \N} (
                    (v \mapsto \semunder{\model}{t_1}\ \nu_i[x \mapsto v])
                )
                \\
                \text{(Definition of semantics)}
                =&\ %
                \bigsqcap\limits_{i \in \N} (\semunder{\model}{\lambda x . t'}\ \nu_i) \ .
            \end{align*}
    \end{enumerate}
\end{proof}

\subsection{Substitution Lemma}
\label{sec:substitution-proof}

Since we have not syntactically defined the evaluation of a $\lambda$-term, our development will need a simple substitution lemma.

\begin{lemma}
\label{lemma:substitution-conc-games}
    For all $\nu : N\disunion V\prightarrow \domain$, we have
    $
        \semunder{\model}{(\lambda x . t)\ t'}\ \nu
        =
        \semunder{\model}{t[x \mapsto t']}\ {\nu}
    $.
\end{lemma}
\begin{proof}
    We show that for all $\nu : N\disunion V\prightarrow \domain$ and all suitable terms $t, t'$, we have
    \[
        \semunder{\model}{(\lambda x . t)\ t'}\ \nu
        =
        \semunder{\model}{t[x \mapsto t']}\ {\nu}
        \ .
    \]
    We have by definition
    \[
        \semunder{\model}{(\lambda x . t)\ t'}\ {\nu}
        =
        (\semunder{\model}{(\lambda x . t)}\ {\nu})\ (\semunder{\model}{t'}\ {\nu})
        =
        \semunder{\model}{t}\ (\nu[x \mapsto \semunder{\model}{t'}\ {\nu}])
    \]
    and show by induction over $t$ that
    \[
        \semunder{\model}{t}\ \nu[x \mapsto \semunder{\model}{t'}\ {\nu}]
        =
        \semunder{\model}{t[x \mapsto t']}\ \nu \ .
    \]
    In the base cases we have
    \begin{enumerate}
        \item
            $
                \semunder{\model}{F}\ {\nu[x \mapsto \semunder{\model}{t'}\ {\nu}]}
                =
                \left( \nu[x \mapsto \semunder{\model}{t'}\ {\nu}] \right) (F)
                =
                \nu (F)
                =
                \semunder{\model}{F}\ {\nu}
                =
                \semunder{\model}{F[x \mapsto t']}\ {\nu}
            $,
        \item
            $
                \semunder{\model}{s}\ {\nu[x \mapsto \semunder{\model}{t'}\ {\nu}]}
                =
                \inter (s)
                =
                \semunder{\model}{s}\ {\nu}
                =
                \semunder{\model}{s[x \mapsto t']}\ {\nu}
            $,
        \item
            $
                \semunder{\model}{x}\ {\nu[x \mapsto \semunder{\model}{t'}\ {\nu}]}
                =
                \semunder{\model}{t'}\ {\nu}
                =
                \semunder{\model}{x[x \mapsto t']}\ {\nu}
            $, and

        \item
            $
                \semunder{\model}{y}\ {\nu[x \mapsto \semunder{\model}{t'}\ {\nu}]}
                =
                \left( \nu[x \mapsto \semunder{\model}{t'}\ {\nu}] \right) (y)
                =
                \nu (y)
                =
                \semunder{\model}{y}\ {\nu}
                =
                \semunder{\model}{y[x \mapsto t']}\ {\nu}
            $,
            for variable $y \neq x$.
    \end{enumerate}
    Then, for the induction step, we first consider application.  That is
    \[
        \semunder{\model}{t_1\ t_2}\ {\nu[x \mapsto \semunder{\model}{t'}\ {\nu}]}
        =
        \left(
        \semunder{\model}{t_1}\ {\nu[x \mapsto \semunder{\model}{t'}\ {\nu}]}
        \right)
        \
        \left(
        \semunder{\model}{t_2}\ {\nu[x \mapsto \semunder{\model}{t'}\ {\nu}]}
        \right)
    \]
    which is equal to, by induction,
    \[
        \left(
        \semunder{\model}{t_1[x \mapsto t']}\ {\nu}
        \right)
        \
        \left(
        \semunder{\model}{t_2[x \mapsto t']}\ {\nu}
        \right)
        =
        \semunder{\model}{t_1[x \mapsto t']\ t_2[x \mapsto t']}\ {\nu}
        =
        \semunder{\model}{(t_1\ t_2)[x \mapsto t']}\ {\nu} \ .
    \]
    Finally, for abstraction, we can assume by $\alpha$-conversion that $y \neq x$, and we have
    \[
        \semunder{\model}{\lambda y . t_1}\ {\nu[x \mapsto \semunder{\model}{t'}\ {\nu}]}
        =
        v \mapsto
        \semunder{\model}{t_1}\ {\nu[x \mapsto \semunder{\model}{t'}\ {\nu}, y \mapsto v]}
    \]
    which is by induction equal to the function
    \[
        v \mapsto
        \semunder{\model}{t_1[x \mapsto t']}\ {\nu[y \mapsto v]}
        =
        \semunder{\model}{\lambda y . t_1[x \mapsto t']}\ {\nu}
        =
        \semunder{\model}{(\lambda y . t_1)[x \mapsto t']}\ {\nu} \ .
    \]
    Thus, by induction, we have the lemma as required.
\end{proof}

\section{Proofs for Section~\ref{Section:Games}}

\subsection{Proof of Lemma~\ref{Lemma:ContinuityRHSConc}}

\begin{proof}
    We show that for all non-ground terminals $s$, $\interconcof{s}$ is $\meet$-continuous.
    We need to treat the terminals $a \colon \ground \to \ground$ of the original scheme and the terminals $\opof{F}$ that were introduced for the determinisation separately.
    In each case, assume a descending chain of arguments $(x_i)_{i \in \N}$.
    \begin{enumerate}
        \item
            \textsfbf{Case $s = a$.}\\
            Since
            $\interconcof{a} = \prependf{a}$
            and we have
            \[
                \prependf{a} (\bigsqcap\limits_{i \in \N} x_i) =
                \bigsqcap\limits_{i \in \N}(\prependf{a}\ x_i)
            \]
            by definition of $\prependf{a}$, we have the property as required.
        \item
            \textsfbf{Case $s = \opof{F}$.}\\
            We show the property when $\opof{F}$ is owned by $\playere$, and thus interpreted as $\ell$-fold disjunction, conjunction is similar.
            We proceed by induction on the arity $\ell$.
            In the base case $\ell = 1$, $\interconcof{\opof{F}}$ is the identity function that is $\meet$-continuous

            Now assume $\opof{F}$ has arity $\ell+1$, and $\interconcof{\opof{F}} = \ldisj{\ell+1}$ is an $\ell+1$-fold disjunction.
            We have
            \begin{align*}
                &\ \ldisj{\ell+1} (\bigsqcap\limits_{i \in \N} x_i)
                \\
                \text{(Definition of $\ldisj{\ell+1}$)}
                =&\ y_1, \ldots, y_{\ell} \ \mapsto \ (\bigsqcap\limits_{i \in \N} x_i) \vee \ldisj{\ell}\ y_1\ \cdots\ y_\ell
                \\
                \text{(Distributivity (see below))}
                =&\ y_1, \ldots, y_{\ell} \ \mapsto \ \bigsqcap\limits_{i \in \N} (x_i \vee \ldisj{\ell}\ y_1\ \cdots\ y_\ell)
                \\
                \text{(Definition of $\sqcap$ for functions)}
                =&\ \bigsqcap\limits_{i \in \N} (y_1, \ldots, y_{\ell} \ \mapsto \ x_i \vee \ldisj{\ell}\ y_1\ \cdots\ y_\ell)
                \\
                \text{(Definition of $\ldisj{\ell+1}$)}
                =&\ \bigsqcap\limits_{i \in \N} \ldisj{\ell+1} \ x_i
            \end{align*}
            In the above we required $\vee$ to distribute over $\meet$, which can be seen by induction over types.
            In the base case, that $\vee$ distributes over $\meet = \wedge$ is standard.
            For the step case, we have for all $f_i$, $g$, and $v$
            \begin{align*}
                &\ %
                \big((\bigsqcap\limits_{i \in \N} f_i) \vee g\big)\ v
                \\
                \quad \quad \quad \quad
                \text{(Definition of $\vee$ and $\meet$)}
                =&\ %
                \big(\bigsqcap\limits_{i \in \N} (f_i\ v) \big) \vee (g\ v)
                \\
                \text{(Induction)}
                =&\ %
                \bigsqcap\limits_{i \in \N} \big((f_i\ v) \vee (g\ v) \big)
                \\
                \text{(Definition of $\vee$ and $\meet$)}
                =&\ %
                \big(\bigsqcap\limits_{i \in \N}(f_i \vee g) \big)\ v \ .
            \end{align*}
    \end{enumerate}
\end{proof}

\subsection{Proof of Theorem~\ref{Theorem:ConcGood}}
\label{Section:ConcGood-proof}

We are required to show $\solconcof{S}$ is satisfied by $\lang{A}$ iff there is a winning strategy foryer Player $\playere$.
The theorem is shown in the following two lemmas.

First, we introduce some notation.
We write
$\prependf{w}$ for $w=a_1\ldots a_n$ to abbreviate $\prependf{a_1} \circ \cdots \circ \prependf{a_n}$.

\begin{lemma}[Player $\playere$]
\label{lemma:player-e-wins}
    If $\solconcof{S}$ is satisfied by $\lang{A}$ there is a winning strategy for $\playere$.
\end{lemma}

\begin{proof}
    In what follows, whenever we refer to a term $t$, we mean a term built over $N\disunion T$, but not over $\detof{T}$.
    The terminals $\opof{F}$ are excluded because they do not occur in the game, they are only introduced in the determinized scheme.

    We will demonstrate a strategy for $\playere$ that maintains the invariant that the current (variable-free) term $t$ reached is such that
    $\concsem{t}\ \solconc$
    is satisfied by $\lang{A}$.
    All plays are infinite or generate a word $w$.
    Since we maintain
    $\concsem{t}\ {\solconc}$
    is satisfied by $\lang{A}$, if $t$ represents a word $w$, we know $w$ is accepted by $A$ and Player $\playere$ wins the game.

    Initially, we have
    $\concsem{S}\ {\solconc} = \solconcof{S}$
    which is satisfied by $\lang{A}$ by assumption.
    Thus, suppose play reaches a term $t$ such that
    $\concsem{t}\ {\solconc}$
    is satisfied by $\lang{A}$.
    There are two cases.

    In the first case
    $t = a_1\ (\cdots\ (a_n\ \wordend))$
    and let
    $w = a_1 \ldots a_n$.
    Since
    \[
        \concsem{a_1(\cdots(a_n(\wordend)))}\ {\solconc}
        =
        \prepend{a_1\ldots a_n}{\varepsilon}
        =
        w
    \]
    and $w$ is satisfied by $\lang{A}$, we know
    $w \in \lang{A}$
    and Player $\playere$ has won the game.

    In the second case, we have
    $t = a_1(\cdots(a_n(F\ t_1 \cdots t_m)))$.
    By assumption, we know
    \begin{multline*}
            \concsem{a_1(\cdots(a_n(F\ t_1 \cdots t_m)))}\ {\solconc} = \\
            \prepend{a_1 \ldots a_n}{
               (\concsem{F}\ {\solconc})\ %
               (\concsem{t_1}\ {\solconc})\ %
               \cdots
               \ (\concsem{t_m}\ {\solconc})
            }
    \end{multline*}
    is satisfied by $\lang{A}$.
    Let
    $F = e_1$,
    \ldots,
    $F = e_\ell$
    be the rewrite rules for $F$.
    There are two subcases.
    \begin{enumerate}
        \item
            If $F$ is owned by $\playere$, then since
            $\concsem{F} =
             \concsem{e_1} \lor \cdots \lor \concsem{e_\ell}$
            there must exist some $i$ such that
            \[
                \prepend{a_1 \ldots a_n}{
                   (\concsem{e_i}\ {\solconc})\ %
                   (\concsem{t_1}\ {\solconc})\ %
                   \cdots
                   \ (\concsem{t_m}\ {\solconc})
                }
            \]
            is satisfied by $\lang{A}$.
            The strategy of Player $\playere$ is to choose the $\nth{i}$ rewrite rule.

            We need to show the invariant is maintained.
            Let
            $e_i = \lambda x_1, \ldots, x_m . e$.
            We have (using the substitution lemma, Lemma~\ref{lemma:substitution-conc-games}),
            \begin{align*}
                &\prepend{a_1 \ldots a_n}{
                   (\concsem{\lambda x_1, \ldots, x_m . e}\ {\solconc})\ %
                   (\concsem{t_1}\ {\solconc})\ %
                   \cdots
                   \ (\concsem{t_m}\ {\solconc})
                }
                \\
                =\ &\prepend{a_1 \ldots a_n}{
                   \concsem{(\lambda x_1, \ldots, x_m . e)\ t_1\ \ldots\ t_m}\ {\solconc}
                }
                \\
                =\ &\prepend{a_1 \ldots a_n}{
                   \concsem{e[x_1 \mapsto t_1, \ldots, x_m \mapsto t_m]}\ {\solconc}
                }
                \\
                =\ &\concsem{a_1(\cdots(a_n(e[x_1 \mapsto t_1, \ldots, x_m \mapsto t_m])))}\ {\solconc} \ .
            \end{align*}
            Note that the term $a_1(\cdots(a_n(e[x_1 \mapsto t_1, \ldots, x_m \mapsto t_m])))$ is the result of Player $\playere$ rewriting $F$ via $F = e_i$.
            Since the satisfaction by $\lang{A}$ passes through the equalities, Player $\playere$'s move maintains the invariant as required.

        \item
            If $F$ is owned by $\playera$ the argument proceeds as in the previous case.
            The key difference is that we have to show satisfaction is maintained no matter which move $\playera$ chooses.
            However, since in this case
            $\concsem{F} =
             \concsem{e_1} \land \cdots \land \concsem{e_\ell}$
            then for all $i$ we have
            \[
                \prepend{a_1 \ldots a_n}{
                   (\concsem{e_i}\ {\solconc})\ %
                   (\concsem{t_1}\ {\solconc})\ %
                   \cdots
                   \ (\concsem{t_m}\ {\solconc})
                }
            \]
            is satisfied by $\lang{A}$.
            The remainder of the argument is identical.
    \end{enumerate}
\end{proof}
\begin{lemma}[Player $\playera$]
\label{lemma:player-a-wins}
    If $\solconcof{S}$ is not satisfied by $\lang{A}$ there is a winning strategy \mbox{for $\playera$.}
\end{lemma}
\begin{proof}
    In what follows, whenever we refer to a term $t$, we mean a term built over $N\disunion T$, but not over $\detof{T}$.
    The terminals $\opof{F}$ are excluded because they do not occur in the game, they are only introduced in the determinized scheme.

    For
    $\fmla \in \domainconcof{\ground}$
    and a variable-closed term $t$ of kind $\ground$, we define \emph{$\fmla$ to be sound for $t$}, denoted $\fmla \soundfor t$, if for all $w\in T^*$ such that
    $\prepend{w}{\fmla}$
    is not satisfied by $\lang{A}$, Player $\playera$ has a winning strategy from term $w(t)$.
    For $w=\varepsilon$, we set $\prepend{\varepsilon}{\fmla}=\fmla$ and let $\varepsilon(t)=t$.
    We can now restate the lemma as
    \begin{align}
    \label{sound}
        \solconcof{S} \soundfor S
        \ .
    \end{align}
    In particular, since $\solconcof{S}$ is not satisfied by $\lang{A}$ it is the case that
    $\prepend{\varepsilon}{\solconcof{S}}$ is not satisfied.
    This means Player $\playera$ has a winning strategy from
    $\varepsilon(S) = S$.
%

    In general, for
    $\soln \in \domainconcof{\kappa_1 \to \kappa_2}$,
    we will also define
    $\soln \soundfor t$
    for terms of kind
    $\kappa_1 \to \kappa_2$.
    That is, for a variable-closed term $t$ of kind
    $\kappa_1 \to \kappa_2$
    and a function
    $\soln \in \domainconcof{\kappa_1 \to \kappa_2}$,
    we define
    $\soln \soundfor t$
    to hold whenever for all variable-closed terms $t'$ of kind $\kappa_1$ and
    $\soln' \in \domainconcof{\kappa_1}$
    such that
    $\soln' \soundfor t'$
    we have
    $\soln\ \soln' \soundfor t\ t'$:
    \[
        \soln \soundfor t,
        \text{ if }
        \forall \ \soln', t'
        \text{ such that }
        \soln' \soundfor t',
        \text{ we have }
        \soln\ \soln' \soundfor t\ t' \ .
    \]

    Similarly, we need to extend $\soundfor$ to terms $t$ with free variables $\vec{x}=x_1 \ldots x_m$.
    Here, we make the free variables explicit and write $t (\vec{x})$.
    We define for
    $\soln : (V\prightarrow \domainconc) \rightarrow \domainconc$
    that
    $\soln \soundfor t(\vec{x})$
    by requiring that for any variable-closed terms
    $t_1, \ldots, t_m$
    and any
    $\soln_1, \ldots, \soln_m \in \domainconc$
    with $\soln_j \soundfor t_j$ for all $1\leq j\leq m$,
    we have
    $\soln\ \nu
     \soundfor
     t [\forall j \colon x_j \mapsto t_j]$,
    where $\nu$ maps $x_j$ to $\soln_j$.

    We now show the following.
    For every number of iterations $i$ in the fixed-point calculation,
    we have $\concsem{t}\ {\solconc^i} \soundfor t$, for all terms $t$ built over the terminals and non-terminals in the scheme of interest.
    After the induction, we will show that the result holds for the greatest fixed point.
    Note that we have a nested induction:
    the outer induction is along $i$, the inner is along the structure of terms.

    Since we are inducting over non-closed terms, we will have to extend $\solconc^i$ to assign valuations to free-variables.
    Thus we will write $\nu^i$ to denote a valuation such that
    \mbox{$\nu^i(F) = \solconc^i(F)$}
    for any non-terminal $F$.


    \noindent
    \textsfbf{Base case $i$.}
    \\
    In the base case, we have $i=0$ and $\nu^i=\top$ for all non-terminals.
    We proceed by induction on the structure of terms.
    We will emphasize if an argumentation is independent of the iteration count.
    This is the case for all terms except non-terminals.

    \noindent
    \textsfbf{Base case $t$.}
    \\
    The base cases of the inner induction that are independent of the iteration count are the following.
    \begin{enumerate}
    \item
        \textsfbf{Case $t = \wordend$.}\\
        For all $i$, we have
        $\concsem{\wordend}\ {\nu^i} = \varepsilon$.
        Take any word $w$ such that
        $\prepend{w}{\varepsilon}$
        is not satisfied by $\lang{A}$.
        No moves can be made from $w(\varepsilon)$ and Player $\playera$ has won the game.

    \item
        \textsfbf{Case $t = a$.}
        \\
        We again reason over all $i$ and show that
        \[
            \concsem{a}\ {\nu^i}\ \soln{} =
            \prependf{a}(\soln{}) \soundfor a(t)
        \]
        for any variable-closed term
        $t : \ground$
        and any $\soln{}$ so that $\soln{} \soundfor t$.
        Take any word $w$ such that
        $\prepend{w}{\prepend{a}{\soln{}}}$
        is not satisfied by $\lang{A}$.
        It follows that
        $\prepend{wa}{\soln{}}$
        is also not satisfied by $\lang{A}$.
        From
        $\soln{} \soundfor t$,
        Player $\playera$ has a winning strategy from $wa(t)$.
        Since $wa(t) = w(a(t))$ we are done.

    \item
        \textsfbf{Case $t = x$.}
        \\
        For all $i$ and all extensions $\nu^i$ of $\solconc^i$, we have
        \[
            \concsem{x}\ {\nu^i}\ =
            \nu^i(x).
        \]
        Take any $\nu^i(x)=\soln$ and any variable-closed term $t'$ with $\soln \soundfor t'$.
        Then $\nu^i(x) \soundfor x[x \mapsto t']$ is immediate.\\[-0.2cm]
\end{enumerate}
The only base case of the inner induction that depends on the iteration count is $t=F$.
Let $F$ take $m$ arguments and consider variable-closed terms $t_1, \ldots, t_m$ with corresponding $\soln_j$ such that $\soln_j \soundfor t_j$.
We have
\begin{align*}
            \concsem{F}\ {\nu^0}\ \soln_1\ \ldots\ \soln_m =
            \solconcof{F}^{0}\ \soln_1\ \ldots\ \soln_m =
            \ftrue.
\end{align*}
Thus, trivially $\concsem{F}\ {\nu^0} \soundfor F$ since true is never unsatisfied.\\[0.2cm]
\textsfbf{Step case $t$.}
\\
In both cases, the argumentation is independent of the actual iteration count.
Therefore, we give it for a general $i$ rather than for $0$.
\begin{enumerate}
    \item
        \textsfbf{Case $t = t'\ t''$.}
        \\
        Assume we already know that
        \begin{align*}
            \concsem{t'}\ {\nu^i} \soundfor t'
            \quad
            \text{ and }
            \quad
            \concsem{t''}\ {\nu^i} \soundfor t''
            \ .
        \end{align*}
        Our task is to show that
        \begin{align*}
            \concsem{t'\ t''}\ {\nu^i} =
            (\concsem{t'}\ {\nu^i})\ (\concsem{t''}\ {\nu^i})
            \soundfor t'\ t''
            \ .
        \end{align*}
        Let the free variables be $x_1, \ldots, x_n$ and consider $\soln_1 \soundfor t_1$ to $\soln_n\soundfor t_n$.
        Let $\nu^i$ map $x_j$ to $\soln_j$ for all $1\leq j\leq n$.
        By the definition of $\soundfor$ for terms with free variables, we have
        $\concsem{t'}\ {\nu^i} \soundfor t' [\forall j: x_j \mapsto t_j]$
        and
        $\concsem{t''}\ {\nu^i} \soundfor t'' [\forall j: x_j \mapsto t_j]$.
        Then, by the definition of $\soundfor$ for functions, we obtain
        \begin{align*}
            \concsem{t'\ t''}\ {\nu^i} =&\
            (\concsem{t'}\ {\nu^i})\ (\concsem{t''}\ {\nu^i})\\
             \soundfor&\
            (t'[\forall j: x_j \mapsto t_j])\ (t''[\forall j: x_j \mapsto t_j]) =
            (t'\ t'')[\forall j: x_j \mapsto t_j] \ .
        \end{align*}
        This means
        $\concsem{t'\ t''}\ {\nu^i} \soundfor t'\ t''$
        as required.

    \item
        \textsfbf{Case $t = \lambda x.e$.}
        \\
        Let the free variables of $e$ be $x, x_1,\ldots, x_n$.
        For $\concsem{\lambda x.e}\ {\nu^i}\soundfor \lambda x.e$, we have to argue that for any $\soln_1\soundfor t_1$ to $\soln_n\soundfor t_n$ with $\nu^i$ mapping $x_i$ to $\soln_i$ for all $1\leq i\leq n$, we get
        \begin{align*}
            \concsem{\lambda x.e}\ {\nu^i} \soundfor (\lambda x.e)[\forall j: x_j\mapsto t_j] \ .
        \end{align*}
        This in turn means that for any $\soln\soundfor t$, we have to show
        \begin{align*}
            (\concsem{\lambda x.e}\ {\nu^i})\ \soln\soundfor ((\lambda x.e)[\forall j: x_j\mapsto t_j])\ t \ .
        \end{align*}
        By the definition of the semantics, we have
        \begin{align*}
            (\concsem{\lambda x.e}\ {\nu^i})\ \soln = \concsem{e}\ {\nu^i}[x\mapsto \soln] \ .
        \end{align*}
        Moreover, since the $t_j$ are variable-closed, they in particular are not affected by replacing $x$ and we get
        \begin{align*}
        ((\lambda x.e)[\forall j: x_j\mapsto t_j])\ t = (\lambda x.(e[\forall j: x_j\mapsto t_j]))\ t.
        \end{align*}
        In the game, $\lambda$-redexes of the form $(\lambda x.e)\ t$ do not occur at all:
        When a non-terminal $F$ is rewritten to its right-hand side $\lambda x.e$, this yields $e[x\mapsto t]$ within a single step.
        This means the game equates $(\lambda x.(e[\forall j: x_j\mapsto t_j]))\ t$ with $e[x\mapsto t, \forall j: x_j\mapsto t_j]$.
        Hence, all that remains to be shown is
        \begin{align*}
            \concsem{e}\ {\nu^i}[x\mapsto \soln] \soundfor e[x\mapsto t, \forall j: x_j\mapsto t_j] \ .
        \end{align*}
        This holds by the hypothesis of the inner induction, showing $\concsem{e}\ {\nu^i}\soundfor e$.
\end{enumerate}

\noindent
\textsfbf{Step case $i$.}
\\
We again do an induction along the structure of terms.
The only case that has not been treated in full generality is $F$.
We now show that $\concsem{F}\ {\nu^{i+1}}\soundfor F$.
Let $F$ take $m$ arguments and consider $\soln_1\soundfor t_1$ to $\soln_m\soundfor t_m$.
The task is to prove $\concsem{F}\ {\nu^{i+1}}\ \soln_1\ \ldots\ \soln_m\soundfor F\ t_1\ \ldots\ t_m$.
To ease the notation, assume there are two right hand sides $e_1', e_2$ for $F$, \ie we have the rules
$F = \lambda x_1 \ldots \lambda x_m.e_1$ and $F = \lambda x_1 \ldots \lambda x_m.e_2$.
This means the right-hand side in the determinised scheme is
$F = \lambda x_1 \ldots \lambda x_m . (\opof{F}\ e_1\ e_{2})$.
Then,
\begin{align*}
    \concsem{F}\ {\nu^{i+1}}
    &= \nu^{i+1}(F)\\
    &= \concsem{ \lambda x_1 \ldots \lambda x_m . (\opof{F}\ e_1\ e_{2}) }\ {\nu^i}\\
    &= v_1, \ldots, v_m \mapsto \concsem{(\opof{F}\ e_1\ e_{2})[x_1 \mapsto v_1, \ldots, x_m \mapsto v_m]}\ {\nu^i}\\
    &= v_1, \ldots, v_m \mapsto \concsem{(\opof{F}\ e_1[\vec{x} \mapsto \vec{v}]\ e_{2}[\vec{x} \mapsto \vec{v}])}\ {\nu^i}\\
    &= v_1, \ldots, v_m \mapsto \interconcof{\opof{F}}\
    \left(
    \concsem{e_1[\vec{x} \mapsto \vec{v}]}\ {\nu^i}
    \right)
    \
    \left(
    \concsem{e_{2}[\vec{x} \mapsto \vec{v}])}\ {\nu^i}
    \right)
\end{align*}
Here, $\interconcof{\opof{F}}$ is a conjunction or disjunction, depending on the owner of $F$.
Recall that the conjunction and disjunction of functions are defined by evaluating the argument functions separately and combining the results.
This means
\begin{align*}
    \concsem{F}\ {\nu^{i+1}}
    &= v_1, \ldots, v_m \mapsto \interconcof{\opof{F}}\
    \left(
    \concsem{e_1[\vec{x} \mapsto \vec{v}]}\ {\nu^i}
    \right)
    \
    \left(
    \concsem{e_{2}[\vec{x} \mapsto \vec{v}])}\ {\nu^i}
    \right)
    \\
    &= v_1, \ldots, v_m \mapsto
    \left(
    \concsem{e_1[\vec{x} \mapsto \vec{v}]}\ {\nu^i}
    \right)
    \ (\vee/\wedge)\
    \left(
    \concsem{e_{2}[\vec{x} \mapsto \vec{v}])}\ {\nu^i}
    \right)
    \\
    &=
    \left(
    v_1, \ldots, v_m \mapsto
    \concsem{e_1[\vec{x} \mapsto \vec{v}]}\ {\nu^i}
    \right)
    \ (\vee/\wedge)\
    \left(
    v_1, \ldots, v_m \mapsto
    \concsem{e_{2}[\vec{x} \mapsto \vec{v}])}\ {\nu^i}
    \right)
    \\
    &=
    \left(
    \concsem{\lambda x_1 \ldots \lambda x_m . e_1}\ {\nu^i}
    \right)
    \ (\vee/\wedge)\
    \left(
    \concsem{\lambda x_1 \ldots \lambda x_m . e_2}\ {\nu^i}
    \right)
    \\
    &=
    \left(
    \concsem{e_1'}\ {\nu^i}
    \right)
    \ (\vee/\wedge)\
    \left(
    \concsem{e_2'}\ {\nu^i}
    \right)
    \ .
\end{align*}
With the same reasoning, we obtain
\begin{align*}
    &\concsem{F}\ {\nu^{i+1}}\ \soln_1\ \ldots\ \soln_m\\
    &\quad = (\concsem{e_1'}\ {\nu^i}\ \soln_1\ \ldots\ \soln_m)\ (\vee/\wedge)\ (\concsem{e_2'}\ {\nu^i}\ \soln_1\ \ldots\ \soln_m).
\end{align*}
We have to prove that for any $w\in T^*$, if $\lang{A}$ does not satisfy the formula
\begin{align*}
    &\ \prepend{w}{(\concsem{e_1'}\ {\nu^i}\ \soln_1\ \ldots\ \soln_m)\ (\vee/\wedge)\ (\concsem{e_2'}\ {\nu^i}\ \soln_1\ \ldots\ \soln_m)}\\
    =& \ \prepend{w}{\concsem{e_1'}\ {\nu^i}\ \soln_1\ \ldots\ \soln_m}\ (\vee/\wedge)\ \prepend{w}{\concsem{e_2'}\ {\nu^i}\ \soln_1\ \ldots\ \soln_m},
\end{align*}
then Player~$\playera$ has a winning strategy from $w(F\ t_1 \ldots\ t_m)$.

Assume Player~$\playere$ owns $F$ and the formula is not satisfied.
If Player~$\playera$ owns $F$, the reasoning is similar.
Since we have a disjunction for Player~$\playere$, $\prepend{w}{\concsem{e_1'}\ {\nu^i}\ \soln_1\ \ldots\ \soln_m}$ is not satisfied.
By the hypothesis of the outer induction, we obtain $\concsem{e_1'}\ {\nu^i}\ \soundfor e_1$ and thus $\concsem{e_1'}\ {\nu^i}\ \soln_1\ \ldots\ \soln_m \soundfor e_1'\ t_1\ \ldots\ t_m$.
As in the case of $\lambda$-abstraction above, we use that the game identifies $e_1'\ t_1\ \ldots\ t_m$ and $e_1 [x_1 \mapsto t_1, \ldots e_m \mapsto t_m]$.
Hence, Player $\playera$ has a winning strategy from $w(e_1 [x_1 \mapsto t_1, \ldots e_m \mapsto t_m])$.
The same argumentation applies to $\prepend{w}{\concsem{e_2}\ {\nu^i}\ \soln_1\ \ldots\ \soln_m}$.
Consequently, whichever move Player $\playere$ makes at
$w(F\ t_1\ \ldots\ t_m)$,
Player $\playera$ has a winning strategy.
\\

\noindent
This finishes the outer induction, proving that $\concsem{t}\ {\solconc^i} \soundfor t$ for all terms $t$ and all $i \in \N$.
We would like to conclude $\concsem{t}\ {\solconc} \soundfor t$.
Since the cppo under consideration is not finite, this needs to be proven separately.\\

\noindent
\textsfbf{Limit case.}
\\
We have shown
$\concsem{t}\ {\solconc^i} \soundfor t$
for all $i \in \N$; we now show
$\concsem{t}\ {\solconc} \soundfor t$
noting by Kleene that
$\solconc = \bigsqcap_{i \in \N} \solconc^i$.
Once we have this we have
$\solconcof{S} \soundfor S$
which proves the lemma.

We formulate a slightly more general induction hypothesis for induction over kinds:
Given a descending sequence of $\soln_i$ for all $i \in \N$ such that each $\soln_i \soundfor t$, we have
$\bigsqcap_{i \in \N} \soln_i \soundfor t$.
In the base case we have $t$ is of kind $\ground$ and we assume
$\soln_i \soundfor t$.
We now argue
$\bigsqcap_{i \in \N} \soln_i \soundfor t$.

Take any $w$ and suppose
$\prepend{w}{\bigsqcap_{i \in \N} \soln_i}$
is not satisfied, then we need to show by the definition of $\soundfor$ that Player $\playera$ has a winning strategy.
Since $\meet$ is conjunction, if
\[
    \prepend{w}{\bigsqcap\limits_{i \in \N} \soln_i} =
    \bigsqcap\limits_{i \in \N} \prepend{w}{\soln_i}
\]
is not satisfied, it must be the case that for some $i$ we have
$\prepend{w}{\soln_i}$
is not satisfied.
In this case, we have $\soln_i \soundfor t$ by assumption and thus by the definition of $\soundfor$ that Player $\playera$ has a winning strategy from $w(t)$.
This proves $\bigsqcap_{i \in \N} \soln_i \soundfor t$.

If $t$ is of kind
$\kappa_1 \to \kappa_2$
we need to show for all
$\soln \soundfor t'$
that
$(\bigsqcap_{i \in \N} \soln_i)\ \soln \soundfor t\ t'$.
We have by the definition of $\meet$ over functions
\[
    (\bigsqcap\limits_{i \in \N} \soln_i)\ \soln
    =
    \bigsqcap\limits_{i \in \N} (\soln_i\ \soln)
\]
Since by assumption on $\soln_i$ and definition of $\soundfor$ for function kinds, we have
$\soln_i\ \soln \soundfor t\ t'$
for each~$i$.
By the induction on the kind, we obtain
$\bigsqcap_{i \in \N}(\soln_i\ \soln) \soundfor t\ t'$
.
Since
$(\bigsqcap_{i \in \N} \soln_i)\ \sol{}
 =
 \bigsqcap_{i \in \N} (\soln_i\ \sol{})$
we establish the desired statement that finishes the induction.

Finally, since
$\concsem{t}\ {\solconc^i}$
satisfies the conditions of the above induction hypothesis and because we have already shown
$\concsem{t}\ {\solconc^i} \soundfor t$
for all $t$, we obtain
\[
    \bigsqcap\limits_{i \in \N} (\concsem{t}\ {\solconc^i}) \soundfor t
    \ .
\]
Then, since using continuity of $\concsem{t}$ we have
\[
    \concsem{t}\ {\solconc} =
    \concsem{t}\ (\bigsqcap\limits_{i \in \N} {\solconc^i}) =
    \bigsqcap\limits_{i \in \N} (\concsem{t}\ {\solconc^i})
\]
we obtain the lemma as required.
\end{proof}


\section{Proofs for Section~\ref{Section:Framework}}
\label{Section:Framework-proofs}

\subsection{Generalising Precision Properties to Functions}
\label{Section:LiftingSurjectivity-proof}

We show that several properties needed for precision can be lifted from the ground domain to function domains.

\begin{lemma}\label{Lemma:LiftingSurjectivity}
    If $\precisionsur$ holds, then for every $\kappa\in \Kappa$ and every $\valr\in\domainrof{\kappa}$ there is a compatible $\vall\in\domainlof{\kappa}$ with $\absof{\vall}=\valr$.
\end{lemma}
\begin{proof}
    We show that, if $\precisionsur$ holds, then for every $\kappa\in \Kappa$ and every $\valr\in\domainrof{\kappa}$ there is a compatible $\vall\in\domainlof{\kappa}$ with $\absof{\vall}=\valr$. 
    We proceed by induction on kinds.
    The base case is given by the assumption (P1) and the fact that every ground element is compatible.
    Assume we have the required surjectivity of $\abs$ for $\kappa_1$ and $\kappa_2$ and consider $\fr\in\domainrof{\kappa_1\rightarrow \kappa_2}$.
    The task is to find a compatible function $\fl$ so that $\absof{\fl}=\fr$.
    Assume $\fr\ \valr = \valr'$.
    By surjectivity for $\kappa_1$, there are compatible elements in $\abs^{-1}(\valr)$, and similar for $\valr'$.
    Let $\vall'$ be a compatible element that is mapped to $\valr'$ by $\abs$.
    We define $\fl\ \vall = \vall'$ for all compatible $\vall \in \abs^{-1}(\valr)$.
    Since $\abs$ is total on $\domainof{\kappa_1}$, this assigns a value to all compatible $\vall$.
    We do not impose any requirements on how to map elements that are not compatible.

    We argue that $\fl$ is compatible. To this end, consider compatible $\vall^1$ and $\vall^2$ with $\absof{\vall^1}=\absof{\vall^2}$.
    By definition, both are mapped identically by $\fl$, $\fl\ \vall^1=\fl\ \vall^2$.
    Hence, in particular the abstractions coincide.
    Moreover, given a compatible $\vall$, we defined $\fl\ \vall = \vall'$ to be a compatible element.

    Concerning the equality of the functions, we have $\absof{\fl}\ \valr =
        \absof{\fl\ \vall} =
        \absof{\vall'} = \valr'$.
    The first equality is the definition of abstraction for functions and the fact that
    $\abs^{-1}(\valr)$ contains compatible elements, one of them being $\vall$, the second is the fact that $\vall$ is mapped to $\vall'$, and the last is by
    $\vall' \in \abs^{-1}(\valr')$.
\end{proof}

\begin{lemma}\label{Lemma:LiftingContinuity}
    If $\precisionsur$ and $\precisioncont$ hold, then for all $\kappa\in \Kappa$ and all descending chains of compatible elements $(f_i)_{i \in \N}$ in $\domainof{\kappa}$, we have $\bigmeet_{i\in\N} f_i$ compatible and $\absof{\bigsqcap_{i\in \N}f_i}=\bigsqcap_{i\in \N}\absof{f_i}$.
\end{lemma}
\begin{proof}
    We proceed by induction on kinds to show that, if $\precisionsur$ and $\precisioncont$ hold, then for all kinds $\kappa\in \Kappa$ and for all descending chains of compatible values $f_1, f_2, \ldots \in\domainof{\kappa}$, we have
    $\bigsqcap_{i\in \N}f_i$ again compatible and $\absof{\bigsqcap_{i\in \N}f_i}=\bigsqcap_{i\in \N}\absof{f_i}$.
    The base case is the assumption.

    In the induction step, let $\kappa=\kappa_1\rightarrow \kappa_2$ and $f_1, f_2,\ldots\in\domainlof{\kappa}$ be a descending chain of compatible elements.
    Let $\vall\in \domainlof{\kappa_1}$ be compatible. 
    The following equalities will be helpful:
    \begin{align*}
    \absof{(\bigmeet_{i\in \N} f_i)\ \vall}=\absof{\bigmeet_{i\in \N} (f_i\ \vall)}=\bigmeet_{i\in \N} \absof{f_i\ \vall} = \bigmeet_{i\in \N} (\absof{f_i}\ \absof{\vall})=(\bigmeet_{i\in \N} \absof{f_i})\ \absof{\vall}.
    \end{align*}
    The first equality is the definition of $\meet$ on functions, the second is the induction hypothesis for $\kappa_2$, the third is compatibility of the $f_i$ and $\vall$, the last is again $\meet$ on functions.

To show compatibility, note that the above implies $\absof{(\bigmeet_{i\in \N} f_i)\ \vall}=\absof{(\bigmeet_{i\in \N} f_i)\ \vall'}$ as long as $\absof{\vall}=\absof{\vall'}$, for all compatible $\vall, \vall'\in\domainlof{\kappa_1}$. 
For compatibility of $(\bigmeet_{i\in \N} f_i)\ \vall$ with $\vall\in\domainlof{\kappa_1}$ compatible, note that $(\bigmeet_{i\in \N} f_i)\ \vall=\bigmeet_{i\in \N} (f_i\ \vall)$. The latter is the meet over a descending chain of compatible elements in $\kappa_2$. By the induction hypothesis on $\kappa_2$, it is again compatible. 

For $\meet$-continuity, consider a value $\valr\in\domainrof{\kappa_1}$. 
By Lemma~\ref{Lemma:LiftingSurjectivity}, there is a compatible $\vall\in\domainlof{\kappa_1}$ with $\absof{\vall}=\valr$.
    We have
    \begin{align*}
        \absof{\bigsqcap_{i\in \N}f_i}\ \valr
        =\absof{(\bigsqcap_{i\in \N}f_i)\ \vall}= (\bigmeet_{i\in \N} \absof{f_i})\ \absof{\vall}=(\bigmeet_{i\in \N} \absof{f_i})\ \valr.
        \end{align*}
        The first equality is the definition of abstraction on functions. 
        Note that we need here the fact that $\bigsqcap_{i\in \N}f_i$ is compatible by the induction hypothesis.
        The second equality is the auxiliary one from above.
        The last equality is by $\absof{\vall}=\valr$.
\end{proof}

\begin{lemma}\label{Lemma:LiftingBottom}
    If $\precisionbot$ holds,
    then
    $\absof{\toplof{\kappa}}=\toprof{\kappa}$
    for all
    $\kappa\in \Kappa$.
\end{lemma}

\begin{proof}
    We show that, if $\precisionbot$ holds,
    then
    $\absof{\toplof{\kappa}}=\toprof{\kappa}$
    for all
    $\kappa\in \Kappa$.
    We proceed by induction on kinds.
    The base case is given by the assumption (P3).
    Assume for $\kappa_2$, we have
    $\absof{\toplof{\kappa_2}}=\toprof{\kappa_2}$.
    Consider function
    $\toplof{\kappa_1\rightarrow \kappa_2}\in\domainlof{\kappa_1\rightarrow\kappa_{2}}$.
    We have to show
    $\absof{\toplof{\kappa_1\rightarrow \kappa_2}}=\toprof{\kappa_1\rightarrow \kappa_2}$.
    If the given top element is not compatible, this holds.
    Assume it is. 
    For
    $\valr \in \domainrof{\kappa_1}$, there are two cases.
    If there is no compatible $\vall\in \domainlof{\kappa_1}$ with $\absof{\vall}=\valr$, we have
    \begin{align*}
        \absof{\toplof{\kappa_1\rightarrow \kappa_2}}\ \valr 
        = \toprof{\kappa_2}
        = \toprof{\kappa_1\rightarrow \kappa_2}\ \valr.
    \end{align*}
    If there is such a $\vall$, we obtain
    \begin{align*}
        \absof{\toplof{\kappa_1\rightarrow \kappa_2}}\ \valr =\absof{\toplof{\kappa_1\rightarrow \kappa_2}\ \vall} = \absof{\toplof{\kappa_2}}=\toprof{\kappa_2} = \toprof{\kappa_1\rightarrow \kappa_2}\ \valr.
    \end{align*}
    The first equality is the definition of abstraction for functions, the next is the fact that
    $\toplof{\kappa_1\rightarrow \kappa_2}$
    maps every element
    $\vall \in \domainlof{\kappa_1}$
    to $\toplof{\kappa_2}$.
    The image of $\toplof{\kappa_2}$ is $\toprof{\kappa_2}$ by the induction hypothesis.
    The last equality is the definition of $\toprof{\kappa_1\rightarrow \kappa_2}$.
\end{proof}

\subsection{Proof of Lemma~\ref{Lemma:MainTransfer}}
\label{Section:MainTransfer-proof}

\begin{proof}
    \newcommand\proofcase[1]{Case #1:\ }
    Assume $\precisionsur$, $\precisionrel$, and $\precisioncomp$ hold. 
    We show, for all terms $t$ and all compatible $\nu$, $\semunder{\modell}{t}\ \nu$ is compatible and $\absof{\semunder{\modell}{t}\ \nu}=\semunder{\modelr}{t}\ \absof{\nu}$. 
    We proceed by structural induction on $t$.
    
    \begin{enumerate}
        \item
            \textsfbf{Case $F$, $x$.}
            \\
            By the assumption, $\semunder{\modell}{F}\ \nu = \nu(F)$ is compatible. 
            Moreover,
            \[
                \absof{\semunder{\modell}{F}\ \nu}=\absof{\nu(F)}=\absof{\nu}(F)=\semunder{\modelr}{F}\ \absof{\nu}
            \]
            holds.
            For $x\in V$, the reasoning is similar.\\[0.2cm]
        \item
            \textsfbf{Case terminal $s$.}
            \\
            Note that $\semunder{\modell}{s}\ \nu=\interlof{s}$. 
            If $s$ is ground, the claim holds by $\precisionrel$.
            Let $s:\kappa_1\rightarrow \kappa_2$. 
            For compatibility, consider $\vall, \vall'\in\domainof{\kappa_1}$ compatible with  $\absof{\vall}=\absof{\vall'}$. 
            Then 
            \begin{align*}
            \absof{\interlof{s}\ \vall}=\interrof{s}\ \absof{\vall} = \interrof{s}\ \absof{\vall'} = \absof{\interlof{s}\ \vall'}.
            \end{align*}
            The first equality is $\precisionrel$, the next is $\absof{\vall}=\absof{\vall'}$, and the last is again $\precisionrel$. 
            The second requirement on compatibility is satisfied by $\precisioncomp$.\\[0.2cm]
             To show $\absof{\semunder{\modell}{s}\ \nu} = \semunder{\modelr}{s}\ \absof{\nu}$, consider a value $\valr \in\domainrof{\kappa_1}$.
            By Lemma~\ref{Lemma:LiftingSurjectivity}, there is some compatible $\vall\in\domainlof{\kappa_1}$ with $\absof{\vall}=\valr$.  
            We have
            \begin{align*}
                    \absof{\interlof{s}}\ \valr =
                \absof{\interlof{s}\ \vall}=
                \interrof{s}\ \absof{\vall}
                =\interrof{s}\ \valr.
            \end{align*}
            The first equality is compatibility of $\interlof{s}$ and the definition of function abstraction. 
            The next equality is $\precisionrel$.
            The last is $\absof{\vall}=\valr$.\\[0.4cm]
    \end{enumerate}
         
    For the induction step, assume the claim holds for $t_1$ and $t_2$.
        
    \begin{enumerate}
        \item 
            \textsfbf{Case $t_1\ t_2$.}
            \\
            For compatibility, observe that
            $\semunder{\modell}{t_1\ t_2}\ \nu=(\semunder{\modell}{t_1}\ \nu)\ (\semunder{\modell}{t_2}\ \nu)$.
            Moreover, $\semunder{\modell}{t_1}\ \nu$ and $\semunder{\modell}{t_2}\ \nu$ are both compatible by the induction hypothesis.
            By definition of compatibility, applying a compatible function to a compatible argument yields a compatible value.
            Hence, $\semunder{\modell}{t_1\ t_2}\ \nu$ is compatible.\\[0.2cm]
            For the equality, note that
            \begin{align*}
                \semunder{\modelr}{t_1\ t_2}\ \absof{\nu} =
                (\semunder{\modelr}{t_1}\ \absof{\nu})\ (\semunder{\modelr}{t_2}\ \absof{\nu}) =
                \absof{\semunder{\modell}{t_1}\ \nu}\ \absof{\semunder{\modell}{t_2}\ \nu}.
            \end{align*}
            The first equality is by the definition of the semantics, the second is the induction hypothesis.
            Compatibility justifies the first of the following equalities.
            The second is again the definition of the semantics:
            \begin{align*}
                \absof{\semunder{\modell}{t_1}\ \nu}\ \absof{\semunder{\modell}{t_2}\ \nu} =
                \absof{(\semunder{\modell}{t_1}\ \nu)\ (\semunder{\modell}{t_2}\ \nu)} =
                \absof{\semunder{\modell}{t_1\ t_2}\ \nu}.
            \end{align*}
            
        \item
            \textsfbf{Case $\lambda x:\kappa.t_1$.}
            \\
            We argue for compatibility.
            Consider compatible $\vall$ and $\vall'$ with $\absof{\vall}=\absof{\vall'}$.
            By definition of the semantics and the induction hypothesis, we have
            \begin{align*}
                \absof{(\semunder{\modell}{\lambda x.t_1}\ \nu)\ \vall} =
                \absof{\semunder{\modell}{t_1}\ \nu[x\mapsto \vall]} =
                \semunder{\modelr}{t_1}\ \absof{\nu[x\mapsto \vall]} \ .
            \end{align*}
            For $\vall'$, the reasoning is similar.
            Since
            $\absof{\vall}=\absof{\vall'}$,
            we have
            $\absof{\nu[x\mapsto \vall]}=\absof{\nu[x\mapsto \vall']}$.
            Hence,
            $\semunder{\modelr}{t_1}\ \absof{\nu[x\mapsto \vall]}=\semunder{\modelr}{t_1}\ \absof{\nu[x\mapsto \vall']}$.
            We conclude the desired equality.
        
            For the second requirement in compatibility, let $\vall$ be compatible.
            By definition of the semantics, $(\semunder{\modell}{\lambda x.t_1}\ \nu)\ \vall = \semunder{\modell}{t_1}\ \nu[x\mapsto \vall]$.
            Since $\nu$ and $\vall$ are compatible, $\nu[x\mapsto \vall]$ is compatible.
            Hence, $\semunder{\modell}{t_1}\ \nu[x\mapsto \vall]$ is compatible by the induction hypothesis.\\[0.2cm]
            To prove
            $\semunder{\modelr}{\lambda x.t_1}\ \absof{\nu} = \absof{\semunder{\modell}{\lambda x.t_1}\ \nu}$, consider an arbitrary value $\valr \in \domainrof{\kappa}$. 
            Let $\vall\in\domainlof{\kappa_1}$ be compatible with $\absof{\vall}=\valr$, which exists by Lemma~\ref{Lemma:LiftingSurjectivity}. 
            We have:
            \begin{align*}
                (\semunder{\modelr}{\lambda x.t_1}\ \absof{\nu})\ \valr =
                \semunder{\modelr}{t_1}\ \absof{\nu}[x \mapsto \valr] =
                \semunder{\modelr}{t_1}\ \absof{\nu[x\mapsto \vall]} \ .
            \end{align*}
            We showed above that $\semunder{\modell}{\lambda x.t_1}\ \nu$ is compatible. 
            Using the definition of abstraction for functions and the definition of the semantics, the other function yields
            \begin{align*}
                \absof{\semunder{\modell}{\lambda x.t_1}\ \nu}\ \valr = \absof{(\semunder{\modell}{\lambda x.t_1}\ \nu)\ \vall}= \absof{\semunder{\modell}{t_1}\ \nu[x\mapsto \vall]} \ .
            \end{align*}
            With the induction hypothesis, 
            $\absof{\semunder{\modell}{t_1}\ \nu[x\mapsto \vall]}=\semunder{\modelr}{t_1}\ \absof{\nu[x\mapsto \vall]}$. 
    \end{enumerate}
\end{proof}

\subsection{Proof of Theorem~\ref{Theorem:FixedPointTransfer}}

\begin{proof}
    Recall $\solfl^0$ and $\solfr^0$ are the greatest elements of the respective domains.
    We have
    \begin{align*}
    \absof{\solfl}=\absof{\bigsqcap_{i\in \N} \rhsfunl^{i}(\solfl^0)}= \bigsqcap_{i\in\N}
    \absof{\rhsfunl^{i}(\solfl^0)} = \bigsqcap_{i\in\N}
    \rhsfunr^{i}(\solfr^0) = \solfr.
    \end{align*}
    The first equality is Kleene's theorem. 
    The second equality uses the fact that each $\rhsfunl^{i}(\solfl^0)$ is compatible and that they form a descending chain (both by induction on~$i$), and then applies Lemma~\ref{Lemma:LiftingContinuity}.
    The third equality also relies on compatibility of the $\rhsfunl^{i}(\solfl^0)$ and invokes Lemma~\ref{Lemma:MainTransfer}.
    Moreover, it needs $\absof{\solfl^0}=\solfr^0$ by Lemma~\ref{Lemma:LiftingBottom}.
    The last equality is again Kleene's theorem.
\end{proof}


\section{Proofs for Section~\ref{Section:Instantiation}}

\subsection{Proof of Lemma~\ref{lem:abs-defined}}

\begin{proof}
    Observe
    $\interabsof{\wordend} = Q_f = \acceptfrom{\varepsilon}$.
    Given a formula $\fmla\in \boolof{\acceptfrom{T^*}}$, we have to show that $\predecf{a}(\fmla)\in  \boolof{\acceptfrom{T^*}}$.
    Since $\predecf{a}$ distributes over conjunction and disjunction, it is sufficient to show the requirement for atomic propositions.
    Consider $Q=\acceptfrom{w}$. We have
    $\interabsof{a}\ \acceptfrom{w} = \predecf{a}(\acceptfrom{w}) = \acceptfrom{a.w}$.
    Finally,  $\interabsof{\opof{F}}$ with $F\in N$ is conjunction or disjunction, and there is nothing to do as the formula structure is not modified.
\end{proof}

\subsection{Proof of Lemma~\ref{lem:abs-continuous}}
\label{Section:abs-continuous-proof}

\begin{proof}
    We require, for all terminals $s$, $\interabsof{s}$ is $\meet$-continuous over the respective lattices.
    We remark that the case $s = \opof{F}$ is identical to Lemma~\ref{Lemma:ContinuityRHSConc}.
    Hence, we show the case $s = a \in \Gamma$.
    Given a descending chain $(x_i)_{i\in\N}$, we have to show
    $\interof{a}\ (\bigsqcap_{i\in\N} x_i) =
     \bigsqcap_{i\in\N} (\interof{a}\ x_i)$.
    Recall that the meet of formulas is conjunction, and that we are in a finite domain.
    The latter means that the infinite conjunction is really the conjunction of finitely many formulas.
    Now $\predecf{a}$ is defined to distribute over finite conjunctions.
    We have
    \[
        \interof{a}\ (\bigsqcap\limits_{i\in\N} x_i)
        =
        \predec{a}{\bigwedge\limits_{i\text{ finite}} x_i}
        =
        \bigwedge\limits_{i\text{ finite}}\predec{a}{x_i}
        =
        \bigsqcap\limits_{i\in\N} (\interof{a}\ x_i)
    \]
    as required.
\end{proof}

\subsection{Proof of Proposition~\ref{Proposition:PrecisionAbs}}
\label{Section:PrecisionAbs-proof}

\begin{proof}
    To show $\abs$ is precise, we have to show $\precisionsur$ to $\precisioncomp$.
    For $\precisionsur$, it is sufficient to argue that for every set of states $Q\in \acceptfrom{T^*}$ there is a word that is mapped to it --- which holds by definition.
    For formulas, note that $\abs=\acceptfromf$ distributes over conjunction and disjunction, which means we can take the same connectives in the concrete as in the abstract and replace the leaves appropriately. Note that we only need a set consisting of one formula.\\[0.2cm]
    $\precisioncont$ is satisfied by the concrete meet being the union of sets of formulas and $\abs$ being defined by an element-wise application.\\[0.2cm]
    For $\precisionbot$, note that the greatest elements are $\set{\ftrue}$ for $\domainconcof{\ground}$ and $\ftrue$ for $\domainabsof{\ground}$. By definition, $\absof{\set{\ftrue}}=\absof{\ftrue}=\ftrue$. \\[0.2cm]
    For $\precisionrel$, consider $\wordend$.
    We have
    $\absof{\interconcof{\wordend}} = \absof{\set{\varepsilon}}=\acceptfrom{\varepsilon}=Q_f=\interabsof{\wordend}$.
    The first equality is by definition of the concrete interpretation, the second is the definition of $\abs$, the third uses the fact that $\varepsilon$ is accepted precisely from the final states, and the last equality is the interpretation of the $\wordend$ in the abstract domain.

    For a letter $a$ and a word $w\subseteq T^*$, we have
    \begin{align*}
    \absof{\interconcof{a}\ w} = \absof{\prependf{a}(w)}= \absof{a.w} = \acceptfrom{a.w}=\predecf{a}(\acceptfrom{w}) = \interabsof{a}\ \absof{w}.
    \end{align*}
    The first equality is the interpretation of $a$ in the concrete, the second is the definition of prepending a letter, the third is the definition of the abstraction, the next is how taking predecessors changes the set of states from which a word is accepted, and the last equality is the interpretation of $a$ in the abstract domain and the definition of the abstraction function.
    The relation generalizes to formulas by noting that both the concrete interpretation and the abstract interpretation of $a$ distribute over conjunction and disjunction.
    It also generalizes to sets of formulas by noting that $\prependf{a}$ is applied to all elements in the set and, in the abstract domain, $\predecf{a}$ distributes over conjunction.

    Let $F$ be a non-terminal owned by $\playera$.
    To simplify the notation, let the associated operation be binary,
        $\opof{F}:\ground\rightarrow \ground \rightarrow \ground$.
    Let $\setfmla_1,\setfmla_2\in\domainconcof{\ground}$ be sets of formulas.
    We have
    \begin{align*}
    \absof{\interconcof{\opof{F}}\ \setfmla_1\ \setfmla_2} = \absof{\setfmla_1\cup \setfmla_2}&= \bigwedge_{\fmla\in\setfmla_1\cup \setfmla_2}\absof{\fmla}\\
    & = \bigwedge_{\fmla\in\setfmla_1}\absof{\fmla}\ \wedge \bigwedge_{\fmla\in\setfmla_2}\absof{\fmla} = \interabsof{\opof{F}}(\absof{\setfmla_1}\ \absof{\setfmla_2}).
    \end{align*}
    The first equality is the concrete interpretation of $\opof{F}$.
    The second is the definition of the abstraction function.
    The third equality holds as we work up to logical equivalence.
    The last is the abstract interpretation of $\opof{F}$ and again the definition of the abstraction.

    Assume $F$ is owned by $\playere$ and $\opof{F}$ is again binary.
    Consider $\setfmla_1,\setfmla_2\in\domainconcof{\ground}$.
    It will be convenient to denote $\Set{\fmla_1\vee\fmla_2}{\fmla_1\in \setfmla_1,\fmla_2\in \setfmla_2}$ by $\setfmla$. We have
   \begin{align*}
    \absof{\interconcof{\opof{F}}\ \setfmla_1\ \setfmla_2} = \absof{\setfmla}&=
    \bigwedge_{\fmla_1\vee\fmla_2\in \setfmla}\absof{\fmla_1\vee\fmla_2}\\
    &=\bigwedge_{\fmla_1\in \setfmla_1, \fmla_2\in \setfmla_2}(\absof{\fmla_1}\vee\absof{\fmla_2})\\
    &= (\bigwedge_{\fmla_1\in\setfmla_1}\absof{\fmla_1})\vee (\bigwedge_{\fmla_2\in\setfmla_2}\absof{\fmla_2})
     = \interabsof{\opof{F}}(\absof{\setfmla_1}\ \absof{\setfmla_2}).
    \end{align*}
    The first equality is the concrete interpretation of $\opof{F}$, the second is the definition of $\abs$ on sets of formulas.
    The third equality is the fact that $\abs$ distributes over disjunctions and rewrites the iteration over the elements of $\setfmla$.
    The following equality is distributivity of conjunction over disjunction, and the fact that we work up to logical equivalence.
    The last is the abstract interpretation of $\opof{F}$ and the definition of the abstraction function.\\[0.2cm]
    It remains to show $\precisioncomp$.
    For $\interconcof{\wordend}$ and $\interconcof{a}$, there is nothing to do as all ground values are compatible.
    Assume $F$ is owned by $\playera$ and $\opof{F}$ is binary.
    The proof for $\playere$ is similar.
    We show that, given a set of formulas $\setfmla$, the function
    $\setfmla\cup -$ is compatible.
    An inspection of the proof of~$\precisionrel$ shows that for any set of formulas $\fmla_1$, we have
    \begin{align*}
    \absof{\setfmla\cup \setfmla_1} = \absof{\setfmla}\wedge \absof{\setfmla_1}.
    \end{align*}
    Hence, if $\absof{\setfmla_1}=\absof{\setfmla_2}$, then $\absof{\setfmla\cup \setfmla_1}=    \absof{\setfmla\cup \setfmla_2}$.
    That $\setfmla\cup \setfmla_1$ is compatible holds as the element is ground.
\end{proof}

\subsection{Proof of Lemma~\ref{Lemma:OptDefined}}

\begin{proof}v
    We have
    $\interoptof{\wordend}=\bigvee Q_f= \absof{Q_f}=\absof{\acceptfrom{\varepsilon}}$.
    For $\interoptof{a}$, we note that both the abstract and the optimized interpretation distribute over conjunctions and disjunctions.
    Hence, it remains to consider whether the application to leaves results in a disjunction that is the image of an abstract set.
    Let $Q=\acceptfrom{w}$.
    We have
    \begin{align*}
    \interoptof{a}\ \absof{\acceptfrom{w}} = \interoptof{a}\ (\bigvee Q)
    &=\bigvee_{q\in Q}\interoptof{a}\ q\\
    &=\bigvee_{q\in Q}\bigvee \predecf{a}(\set{q})\\
    &=\bigvee \predecf{a}(Q)\\
    &=\absof{\predecf{a}(Q)}
    =\absof{\predecf{a}(\acceptfrom{w})}=\absof{\acceptfrom{a.w}}.
    \end{align*}
    The first equality is the definition of the abstraction function.
    Then we apply distributivity of the optimized interpretation of $a$ over disjunctions.
    The following equality is the actual interpretation of $a$ in the optimized model.
    The next equality uses $\predecf{a}(Q)=\bigcup_{q\in Q}\predecf{a}(q)$.
    The following is again the definition of the abstraction function.
    Then we replace $Q$ by its definition.
    Finally, we note the interplay between $\predecf{a}$ and $\acceptfrom{-}$.

    For conjunction and disjunction, which are used as the interpretation of $\opof{F}$ depending on the player, we note that $\abs$ distributes to the arguments.
    Hence, if the arguments are $\absof{\fmla_1}$ and $\fmla_2$, we have
    $\absof{\fmla_1}\wedge \absof{\fmla_2}=\absof{\fmla_1\wedge \fmla_2}$.
\end{proof}

\subsection{Proof of Lemma~\ref{Lemma:OptContinuous}}
\label{Section:OptContinuous-proof}

\begin{proof}
    We need, for all terminals $s$, $\interoptof{s}$ is $\meet$-continuous over the respective lattices.
    We remark that the case $s = \opof{F}$ is identical to Lemma~\ref{Lemma:ContinuityRHSConc}.
    The case $s = a \in \Gamma$ follows from distributivity of $\interoptof{a}$ as in the proof of Lemma~\ref{lem:abs-continuous}.
\end{proof}


\subsection{Proof of Proposition~\ref{Proposition:PrecisionOpt}}
\label{Section:PrecisionOpt-proof}

\begin{proof}
    We show the optimized abstraction is precise.
    Surjectivity in $\precisionsur$ holds by definition as does $\precisionbot$.
    Also $\meet$-continuity in $\precisioncont$ is by the fact that the meets over the concrete domain are finite, and hence the definition of $\abs$ already yields continuity.
    We argue for $\precisionrel$.\\[0.2cm]
    For $\wordend$, Lemma~\ref{Lemma:OptDefined} yields $\interoptof{\wordend}=\absof{Q_f}$, which is $\absof{\interabsof{\wordend}}$ as required.
    For $a$, the same lemma shows $\interoptof{a}\ \absof{\acceptfrom{w}}=\absof{\predecf{a}(\acceptfrom{w})}$, which is $\absof{\interabsof{a}\ \acceptfrom{w}}$.
    The equality generalizes to formulas as both, the abstraction function and the interpretations distribute over conjunctions and disjunctions.
    For $\opof{F}$, assume it is a binary conjunction.
    We have
    \begin{align*}
    \interoptof{\opof{F}}\ \absof{\fmla_1}\ \absof{\fmla_2}=
    \absof{\fmla_1}\wedge \absof{\fmla_2}=\absof{\fmla_1\wedge \fmla_2}=\absof{\interabsof{\opof{F}}\ \fmla_1\ \fmla_2}.
    \end{align*}
    The first equality is the definition of the interpretation in the optimized model, the next is distributivity of $\abs$ over conjunction.
    Finally, we have the interpretation of $\opof{F}$ in the abstract model.\\[0.2cm]
    For $\precisioncomp$, there is nothing to do for
    $\interconcof{\wordend}$ and $\interconcof{a}$, as all ground values are compatible.
    We consider the conjunctions and disjunctions used to resolve the non-determinism.
    Consider a formula $\fmla$.
    The task is to show that the function $\fmla\wedge -$ is compatible.
    Consider $\fmla_1$ and $\fmla_2$ with $\absof{\fmla_1}=\absof{\fmla_2}$.
    Then
    \begin{align*}
    \absof{\fmla\wedge \fmla_1}=\absof{\fmla}\wedge\absof{\fmla_1}=\absof{\fmla}\wedge\absof{\fmla_2}=\absof{\fmla\wedge\fmla_2}.
    \end{align*}
    The first equality is distributivity of the abstraction function over conjunctions.
    The next is the assumed equality.
    The third is again distributivity.
    Compatibility of $\fmla\wedge\fmla_1$ holds as ground values are always compatible.
\end{proof}

\subsection{Proof of Corollary~\ref{Theorem:Complexity}}
\label{Section:Complexity-proof}

To show the complexity, we argue the upper and lower bounds separately.

\begin{proof}[Proof of Proposition~\ref{lem:membership}]
    We need to argue that $\sol{\modelopt}$ can be computed in $(k+1)$-times exponential time.
    We have that $\sol{\modelopt} = \bigsqcap_{i \in  \N} \rhsfunopt^{i}(\solfl^0)$.
    Since the domains $\domainoptof{\kappa}$ are finite for all kinds $\kappa$, there is an index $i_0 \in \N$ such that $\sol{\modelopt} = \bigsqcap_{i = 0}^{i_0} \rhsfunopt^{i}(\solfl^0) = \rhsfunopt^{i_0}(\solfl^0)$.
    In the following, we will see that the number of iterations, \ie the index $i_0$ is at most $(k+1)$-times exponential, and that one iteration can be executed in $(k+1)$-times exponentially many steps.

    First, we reason about the number of iterations.
    For a partial order $\calD$, we define its \emph{height} $h(\calD)$ as the length of the longest strictly descending chain, \ie the height is $m$ if the longest such chain is of the shape
    \[
        x_0 > x_1 > \ldots > x_k
        \ .
    \]
    The height of the domain is an upper bound for $i_0$ by its definition:
    If for some index $i_1$ we have $\rhsfunopt^{i_1}(\solfl^0) = \rhsfunopt^{i_1 + 1}(\solfl^0)$, we know $\bigsqcap_{i = 0}^{i_0} \rhsfunopt^{i}(\solfl^0) = \rhsfunopt^{i_01}(\solfl^0)$ and thus $i_1 = i_0$.
    Such an index $i_1$ has to exist and has to be smaller than the height of the domain, otherwise the sequence of the $\rhsfunopt^{i}(\solfl^0)$ would form a chain that is strictly longer than the height, a contradiction to the definition.

    It remains to see what the height of our optimized domain is.
    Recall that $\rhsfunopt$ has the type signature $(N \to \domainopt) \to (N \to \domainopt)$.
    Our goal in the following is to determine $h(N \to \domainopt )$.
    We can identify $N \to \domainopt$ with $\domainoptof{F_1} \times \ldots \times \domainoptof{F_\ell}$, where $F_1, \ldots, F_\ell$ are the non-terminals of the scheme.
    The height of this product domain is the sum of its height.
    We are done if we show that even the domain $\domainoptof{F}$ with the maximal height is $(k+1)$-times exponentially high, since the number of non-terminals is polynomial in the input scheme.

    In the following we prove:
    If kind $\kappa$ is of order $k'$, then $\domainoptof{\kappa}$ has $(k'+1)$-times exponential height.
    For the induction step, we also need to consider the cardinality of $\domainoptof{\kappa}$,
    therefore, we strengthen the statement and also prove that the cardinality $\card{\domainoptof{\kappa}}$ is $(k'+2)$-times exponential.

    We proceed by induction on $k'$.

    In the base case $k' = 0$, we necessarily have $\kappa = \ground$, and indeed the domain \mbox{$\alpha(\boolof{\acceptfrom{T^*}}) \subseteq \boolof{\QNFA}$} is singly exponentially high.
    To see that this is the case, consider a strictly decreasing chain $(\phi_j)_{j \in \N}$ of positive boolean formulas over $\QNFA$, \ie a chain where each formula is strictly implied by the next.
    To each formula, $\phi_j$, we assign the set $\calQ_j = \Set{Q \subseteq \QNFA}{ Q \text{ satisfies } \phi_j}$ of assignments under which $\phi_j$ evaluates to true.
    That $\phi_j$ is strictly implied by $\phi_{j+1}$ translates to the fact that $\calQ_j$ is a strict subset of $ \calQ_{j+1}$.
    This gives us that the sets $\calQ_j$ themselves form a strictly ascending chain in $\pwrset{\pwrset{\QNFA}}$, and it is easy to see that such a chain has length at most $\card{ \pwrset{\QNFA}  } = 2^{\card{\QNFA}}$.

    Furthermore, we can represent each equivalence class of formulas in $\boolof{\QNFA}$ by a representative in conjunctive normal form, \ie by an element of $\pwrset{\pwrset{\QNFA}}$.
    This shows that the cardinality of the domain is indeed bounded by
    $\card{ \pwrset{\pwrset{\QNFA}} } = 2^{ \card{\pwrset{\QNFA}} } = 2^{2^{\card{\QNFA}}}$.

    Now assume the statement holds for $k'$, and consider $\kappa$ of order $k'+1$.
    We need an inner induction on the arity $m$ of $\kappa$.

    Since $\ground$ is the only kind of arity $0$, and does not have order $k'+1$ for any $k'$, there is nothing to do in the base case.

    Now assume that $\kappa = \kappa_1 \to \kappa_2$.
    By the definitions of arity and order, we know that $\kappa_1$ is of order at most $k'$, therefore we now by the outer induction that the height of $\domainoptof{\kappa_1}$ is at most $(k'+1)$-times exponential.
    The order of $\kappa_2$ is at most $(k'+1)$, but the arity of $\kappa_2$ is strictly less than the arity of $\kappa$, thus we get by the inner induction that the height of $\domainoptof{\kappa_2}$ is at most $(k'+2)$-times exponential.

    The domain $\domainoptof{\kappa_1 \to \kappa_2} = \contfun{\domainoptof{\kappa_1}}{\domainoptof{\kappa_2}}$ is a subset of all functions from $\domainoptof{\kappa_1}$ to $\domainoptof{\kappa_2}$.
    Let us reason about the height of this more general function domain.
    We know that its height is the height of the target times the size of the source, \mbox{\ie $h(\domainoptof{\kappa_2}) \cdot \card{ \domainoptof{\kappa_1} }$.}
    The induction completes the proof, as both $h(\domainoptof{\kappa_2})$ and $\card{ \domainoptof{\kappa_1} }$ are at most $(k'+2)$-times exponential.

    It remains to argue that each iteration can be implemented in at most $(k+1)$-times exponentially many steps.
    To this end, we argue that each element of $\domainoptof{\kappa}$ can be represented by an object of size $(k'+1)$-times exponential, where $k'$ is the order of $\kappa$.
    It is easy to see that all operations that need to be executed on these objects, namely evaluation, conjunction, disjunction, and predecessor computation can be implemented in polynomial time in the size of the objects.

    Let $k' = 0$, \ie $\kappa = \ground$.
    We again represent each element of $\domainoptof{\ground}$ by a formula over $\QNFA$ in conjunctive normal form, \ie as an element of $\pwrset{ \pwrset{ \QNFA}}$.
    In the worst case, one single formula $\phi$ contains everyone of the $2^{\QNFA}$ many clauses, each clause having size at most $\card{\QNFA}$.
    This means that one formula needs at most singly exponential space.

    For the induction step, consider $\kappa$ of order $k+1$.
    As above, we need an inner induction on the arity of $\kappa$, for which the base case is trivial.

    Let $\kappa = \kappa_1 \to \kappa_2$.
    An element of $\domainoptof{\kappa}$ is a function that assigns to each of the $\card{\domainoptof{\kappa_1}}$-many elements of $\domainoptof{\kappa_1}$ an element of $\domainoptof{\kappa_2}$.
    In the previous part of the proof, we have argued, that $\card{\domainoptof{\kappa_1}}$ is at most $(k+1)$ times exponential.
    By the induction on the arity, we know that each object in $\domainoptof{\kappa_2}$ can be represented in at most $(k+2)$-times exponential space.
    This shows that objects of  $\domainoptof{\kappa}$ can be represented using $(k+2)$-times exponential space, and finishes the proof.
\end{proof}


We show that determining the winner in a higher-order word game is $(k+1)\mathsf{EXP}$-hard for an order-$k$ recursion scheme.

\begin{proof}[Proof of Proposition~\ref{lem:hardness}]
    \newcommand\kitpdatape{$k$-PDA$^{+}$\xspace}
    \newcommand\kitpda{$k$-PDA\xspace}
    \newcommand\initcell{\circ}
    We begin with a result due to Engelfriet~\cite{E91} that shows alternating $k$-iterated pushdown automata with a polynomially bounded auxiliary work-tape (\kitpdatape) characterize the $(k+1)\mathsf{EXP}$ word languages.
    We fix any $(k+1)\mathsf{EXP}$-hard language and its corresponding alternating \kitpda $B$.
    Let
    $\lang{B}$
    be the set of words accepted by $B$.
    Deciding if a given word $w$ is in the language defined by $B$ is $(k+1)\mathsf{EXP}$-hard in the size of $w$ (recall $B$ is fixed).
    We show that this problem can be reduced in polynomial time to an inclusion problem
    $\lang{B'} \subseteq \lang{A}$
    for some $k$-iterated pushdown automaton (without work-tape) (\kitpda) $B'$ and NFA $A$ of size polynomial in the length of $w$.
    From $B'$, we can construct in polynomial time an equivalent game over a scheme $G$.
    This will show the game language inclusion problem for order-$k$ schemes is $(k+1)\mathsf{EXP}$-hard.

    In an alternating \kitpdatape, there are two Players $\playere$ and $\playera$.
    When decided whether a word $w$ is in the language of a \kitpdatape, $\playere$ will attempt to prove the word is in the language, while $\playera$ will try to refute it.

    We first describe how to obtain $B'$ from $B$.
    Since the word $w$ is fixed, we can force $B$ to output the word $w$ by forming a product of $w$ with the states of $B$.
    Call this automaton $B \times w$.
    This reduces the word membership problem to the problem of determining whether $B \times w$ can reach an accepting state.
    Next, to remove the worktape from $B \times w$ (and form $B'$) we replace the output of $B \times w$ (which will always be $w$ or empty) with a series of guesses of the worktape.
    That is, a transition of $B \times w$ will be simulated by $B'$ by first making a transition as expected, and then outputting a guess (consistent with the transition) of what the worktape of $B \times w$ should be.
    The automaton $A$ will accept a guessed sequence of worktapes iff it is able to find an error in the sequence.
    The word $w$ will be in the language of $B$ if $B'$ is able to reach a final state and produce a word $w'$ that is correct; that is, $w'$ is \emph{not} in the language of $A$.

    Note, here, the reversal of the roles of the Players.
    In $B$, control states are owned by $\playere$ or $\playera$.
    When determining if
    $w \in \lang{B}$
    for some $w$, the first Player $\playere$ tries to show the word is accepted,
    while the second Player $\playera$ tries to force a non-accepting run.
    In $B'$, however, $w$ is accepted iff the output of $B'$ is not included in the language of $A$.
    Thus, $\playera$ will effectively be aiming to prove that
    $w \in \lang{B}$.

    In more detail, we take any $(k+1)\mathsf{EXP}$-hard language and its equivalent (fixed) alternating \kitpdatape.
    Given a word $w$, deciding
    $w \in \lang{B}$
    is $(k+1)\mathsf{EXP}$-hard.
    We define $B'$ directly from $B$ rather than going through the intermediate $B \times w$.

    A transition
    $(p, a, o, \sigma, p')$
    of $B$ means the following.
    From control state $p$, upon reading a character $a$ from $w$, apply operation $o$ to the work-tape (which may become stuck if not applicable) and operation $\sigma$ to the stack (which may also become stuck if not applicable).
    Next, move to control state $p'$, from which the remainder of $w$ is to be read.

    Let $m$ be the polynomial bound on the size of the work-tape of $A$ given the input word $w$.
    Let $\Sigma$ be the alphabet of the work-tape.
    Let the set of work-tape operations
    $O = \set{o_1, \ldots, o_n}$
    and work-tape positions
    $P = \oneto{m}$
    be disjoint from $\Sigma$.
    Also, let
    $\initcell \in \Sigma$
    be the initial symbol appearing in each cell of the initial work-tape.
    We will construct $A'$ such that
    \[
        \lang{A'} \subseteq \initcell^m \left( P O  \Sigma^m \right)^\ast \ .
    \]
    That is, $A'$ outputs a sequence of work-tape configurations separated by positions in $P$ and operations in $O$.
    That is, $A'$ will simulate a run of $A$ over $w$.

    For every control state $p$ of $A$, we will have control states
    $(p, w')$
    of $A'$, where $w'$ is a suffix of $w$.
    We will also have
    $(p, w', o)$
    where $o$ is a work-tape operation to be applied.
    Then for each transition
    $(p, a, o, \sigma, p')$
    of $B$ we have a transition
    $((p, aw'), \varepsilon, \sigma, (p', w', o))$
    of $B \times w$.
    From
    $(p', w', o)$
    the automaton $B'$ will output some character from $P$
    (a guess at the work-tape head position),
    followed by $o$ (to indicate the operation applied).
    It will then be able to output any word from $\Sigma^m$ (a guess of the work-tape contents) before moving to
    $(p', w')$
    and continuing the simulation.
    Initially, $B'$ will simply output $\initcell^m$ and move to control state
    $(p, w)$
    where $p$ is the initial control state of $B$.

    The final step in defining $B'$ is to assign ownership of the control states.
    Recall, we needed to switch the roles of the Players.
    Thus, we define
    $O((p, w)) = \playere$
    whenever $p$ belongs to $\playera$ in $B$.
    All other control states of $B'$ are owned by $\playera$.
    We define the accepting control states to be those of the form
    $(p, \varepsilon)$
    where $p$ is accepting in $B$.
    Observe these have no outgoing transitions.

    Next we define the regular automaton $A$ which detects mistakes in the work-tape.
    Such an error is either due to a poorly updated cell, or due to a poorly updated head position.
    The set of work-tape operations $O$ is such that there is a mapping
    \[
        \pi : (P \times O \rightarrow P)
              \cup
              (P \times \Sigma \times P \times O \rightarrow \Sigma \cup \{\bot\})
    \]
    where
    $\bot \notin \Sigma$
    and
    \begin{itemize}
    \item
        $\pi(i, o) = j$
        means if the head is at position $i$, it is at position $j$ after operation $o$,
        and
    \item
        $\pi(i, \alpha, j, o) = \beta$
        means, if the head is at position $i$, $\alpha$ is the contents of the cell at position $j$, and operation $o$ is applied, then $\beta$ is the contents of the cell after applying $o$.
        If
        $\beta = \bot$
        then $o$ could not be applied to this work-tape and became stuck.
        (E.g.\ if $i = j$ and the operation required the head to read a character other than $\alpha$.)
    \end{itemize}
    Thus, we require the following regular language, for which a polynomially-sized regular automaton is straightforward to construct.
    Let $\Gamma = \Sigma \cup P \cup O$.
    \[
        \lang{A}
        =
        \left(
            \Gamma^\ast
            \left(
                \bigcup\limits_{\pi(i, o) \neq j}
                i o \Sigma^m j
            \right)
            \Gamma^\ast
        \right)
        \cup
        \left(
            \Gamma^\ast
            \left(
                \bigcup\limits_{\pi(i, \alpha, j, o) \neq \beta}
                    i o \Sigma^j \alpha \Gamma^{m+2} \beta
            \right)
            \Gamma^\ast
        \right) \ .
    \]

    We have thus defined a \kitpda $B'$ that produces some word $w'$ not accepted by $A$ iff $w$ is accepted by $B$.

    The final step is to produce a game over a scheme $G$ that is equivalent to the game problem for $k$-iterated pushdown automata.
    This is in fact a straightforward adaptation of the techniques introduced by Knapik~\textit{et al.}~\cite{KNU02}.
    However, we choose to complete the sketch using definitions from Hague~\textit{et al.}~\cite{HMOS08} as we believe these provide a clearer reference.
    In particular, we adapt their Definition 4.3.

    The key to the reduction is a tight correspondence (given in op.~cit.) between configurations
    $(q, s)$
    of a $k$-iterated pushdown automaton, and terms of the form\footnote{%
        In fact, in op.~cit.~non-terminals had the form
        $F^{a, e}_q \vec{\Psi} \vec{\Psi}_{k-1} \cdots \vec{\Psi}_0$.
        where $e$ and $\vec{\Psi}$ are used to handle \emph{collapse links}, which we do not need here.
    }
    $F^a_q \vec{\Psi}_{k-1} \cdots \vec{\Psi}_0$.
    That is, every configuration is represented (in a precise sense) by such a term and every term of such a form represents a configuration.
    Moreover, for every transition
    $(q, a, o, \sigma, q')$
    of the pushdown automaton, when
    $o \neq \varepsilon$
    we can associate a rewrite rule of the scheme
    \[
        F^a_q  = \lambda \vec{x} . o(e_{(q', \sigma)})
    \]
    such that the term obtained by applying the rewrite rule to
    $F^a_q \vec{\Psi}_{k-1} \cdots \vec{\Psi}_0$
    is a term
    $o(F^b_{q'} \vec{\Psi}'_{k-1} \cdots \vec{\Psi}'_0)$
    where
    $F^b_{q'} \vec{\Psi}'_{k-1} \cdots \vec{\Psi}'_0$
    represents the configuration reached by the transition.
    That is,
    $(q', \sigma(s))$.
    When
    $o = \varepsilon$
    we simply omit $o$, that is
    \[
        F^a_q  = \lambda \vec{x} . e_{(q', \sigma)} \ .
    \]

    To each non-terminal, we assign
    $O(F^a_q) = \playere$
    whenever $q$ is a $\playere$ control state.
    Otherwise,
    $O(F^a_q) = \playera$.
    For every accepting control state $q$ we introduce the additional rule
    \[
        F^a_q = \lambda \vec{x} . \wordend \ .
    \]
    Finally, we have an initial rule
    \[
        S = t
    \]
    where $t$ is the term representing the initial configuration.

    Given the tight correspondence between configurations and transitions of the \kitpda and terms and rewrite steps of $G$, alongside the direct correspondence between the owner of a control state $q$ and the owner of a non-terminal of $G$, it is straightforward to see, via induction over the length of an accepting run in one direction, or derivation sequence in the other, that $B'$ is able to produce a word not in $A$ iff a word not in $A$ is derivable from $S$.
    Thus, we have reduced the word acceptance problem for some alternating \kitpdatape to the game problem for language inclusion of a scheme.
    This shows the problem is
    $(k+1)\mathsf{EXP}$-hard.
\end{proof}

\end{document}